\documentclass{article}
\usepackage[utf8]{inputenc}
\usepackage{mystyle}
\usepackage{hyperref}
\usepackage[capitalize]{cleveref}
\usepackage{crossreftools}

\newcommand{\cutforipco}[1]{{\color{red}#1}}
\renewcommand{\cutforipco}[1]{}
\title{Online Hypergraph Matching}
\date{}
\usepackage{amsmath}
\usepackage{placeins}
\newcommand{\optLP}{\mathsf{OPT_{LP}}}
\newcommand{\opt}{\mathsf{OPT}}
\newcommand{\val}{\mathcal{V}}
\newcommand{\Exp}{\mathbb{E}}
\newcommand{\Var}{\operatorname{Var}}
\newcommand{\Alg}{\mathcal{A}}
\newcommand{\hfunc}{f}
\newcommand{\Match}{\mathcal{M}}
\newcommand{\numact}{r}
\newcommand{\Hyp}{\mathcal{H}}
\usepackage[disable]{todonotes}

\newcommand{\mesh}{\eta}
\newcommand{\1}{\mathbbm 1}
\usepackage{tikz}
\usepackage{enumitem}

\usepackage{authblk}
\newlist{proplist}{enumerate}{1} \setlist[proplist]{left=0pt, itemsep=1pt, label=Property \arabic*., ref=\arabic*}

\newcommand{\colorline}[1]{
    \vspace{-0.04cm}
\hspace{-0.22cm}\colorbox{gray!10}{\makebox[0.99\linewidth][l]{#1}}
}

\author{Sander Borst}
\affil{Max Planck Institute for Informatics, Germany}
\author{Danish Kashaev}
\affil{Centrum Wiskunde \& Informatica, The Netherlands}
\author{Zhuan Khye Koh}
\affil{Boston University, USA}
\affil[ ]{\tt{sborst@mpi-inf.mpg.de,danish.kashaev@cwi.nl,zkkoh@bu.edu}}

\begin{document}

\title{Online Matching on 3-Uniform Hypergraphs\thanks{This project has received funding from the European Research Council (ERC) under the European Union’s Horizon 2020 research and innovation programme (grant agreement no. 805241--QIP). \\ 
S. Borst and Z. K. Koh --- Research was done while the authors were at CWI.}}
\maketitle   

\begin{abstract}
The online matching problem was introduced by Karp, Vazirani and Vazirani (STOC 1990) on bipartite graphs with vertex arrivals.
It is well-known that the optimal competitive ratio is $1-1/e$ for both integral and fractional versions of the problem.
Since then, there has been considerable effort to find optimal competitive ratios for other related settings.

In this work, we go beyond the graph case and study the online matching problem on $k$-uniform hypergraphs.
For $k=3$, we provide an optimal primal-dual fractional algorithm, which achieves a competitive ratio of $(e-1)/(e+1)\approx 0.4621$. 
As our main technical contribution, we present a carefully constructed adversarial instance, which shows that this ratio is in fact optimal. 
It combines ideas from known hard instances for bipartite graphs under the edge-arrival and vertex-arrival models.

For $k\geq 3$, we give a simple integral algorithm which performs better than greedy when the online nodes have bounded degree.
As a corollary, it achieves the optimal competitive ratio of 1/2 on 3-uniform hypergraphs when every online node has degree at most 2.
This is because the special case where every online node has degree 1 is equivalent to the edge-arrival model on graphs, for which an upper bound of 1/2 is known.
\end{abstract}
\section{Introduction}
Online matching is a classic problem in the field of online algorithms.
It was first introduced in the seminal work of Karp, Vazirani and Vazirani~\cite{karp1990optimal}, who considered the bipartite version with one-sided vertex arrivals.
In this setting, we are given a bipartite graph where vertices on one side are known in advance (offline), and vertices on the other side arrive sequentially (online).
When an online vertex arrives, it reveals its incident edges, at which point the algorithm must decide how to match it (or not) irrevocably.
The goal is to maximize the cardinality of the resulting matching.
Karp et al.~\cite{karp1990optimal} gave an elegant randomized algorithm {\sc Ranking}, which achieves the optimal competitive ratio of $1-1/e$.

In certain applications, each offline vertex may be matched more than once.
Examples include matching online jobs to servers, or matching online impressions to advertisers.
This is the online $b$-matching model of Kalyanasundaram and Pruhs~\cite{kalyanasundaram2000optimal}, where $b\geq 1$ is the maximum number of times an offline vertex can be matched.
As $b$ and the number of online vertices tend to infinity, it in turn captures the fractional relaxation of the Karp et al.~\cite{karp1990optimal} model.
This means that the algorithm is allowed to match an online node fractionally to multiple neighbours, as long as the total load on every vertex does not exceed 1.
For this problem, it is known that the deterministic algorithm {\sc Balance} (or {\sc Water-Filling}) achieves the optimal competitive ratio of $1 - 1/e$.

\subsection{Online hypergraph matching}
The online bipartite matching problem can be naturally generalized to hypergraphs as follows.
For $k\geq 2$, let $\Hyp= (V,W,H)$ be a $k$-uniform hypergraph with offline vertices $V$, online vertices $W$ and hyperedges $H$.
Every hyperedge $h\in H$ contains $k-1$ elements from $V$ and 1 element from $W$.
Just like before, the online vertices arrive sequentially with their incident hyperedges, and the goal is to select a large matching, i.e., a set of disjoint hyperedges.
The greedy algorithm is $1/k$-competitive.
On the other hand, no integral algorithm can be $2/k$-competitive \footnote{In \cite{journals/orl/TrobstU24}, it is shown that no algorithm can be $(2+
f(k))/k$-competitive for some positive function $f$ with $f(k)=o(1)$. In Appendix \ref{sec:integral_hardness}, we give a simple construction showing that no integral algorithm can be $2/k$-competitive.}.

For the \emph{fractional} version of the problem, Buchbinder and Naor \cite{buchbinder2009online} gave a deterministic algorithm which is $\Omega(1/\log k)$-competitive.
They also constructed an instance showing that any algorithm is $O(1/\log k)$-competitive.
In fact, their results apply to the more general setting of online packing linear program (LP), in which variables arrive sequentially.
In the context of hypergraphs, this means that the hyperedges arrive sequentially.
Note that for $k$-uniform hypergraphs, there is a trivial reduction from this \emph{edge-arrival} model to our \emph{vertex-arrival} model on $(k+1)$-uniform hypergraphs, by adding degree 1 online nodes.

The aforementioned results show that asymptotically, both integral and fractional versions of the online matching problem on $k$-uniform hypergraphs are essentially settled (up to constant factors).
However, our understanding of the problem for small values of $k$ (other than $k=2$) remains poor.
Many applications of online hypergraph matching in practice have small values of $k$.
For instance, in ride-sharing and on-demand delivery services~\cite{journals/ior/PavoneSST22}, $k-1$ represents the capacity of service vehicles, which is often small.
Another example is network revenue management problems~\cite{journals/ior/MaRST20}.
In this setting, given a collection of limited resources, a sequence of product requests arrive over time.
When a product request arrives, we have to decide whether to accept it irrevocably.
Accepting a product request generates profit, but also consumes a certain amount of each resource.
The goal is to devise a policy which maximizes profit.
In this context, $k-1$ represents the maximum number of resources used by a product.
As Ma et al.~\cite{journals/ior/MaRST20} noted, many of these problems have small values of $k$.
In airlines, for example, $k-1$ corresponds to the maximum number of flight legs included in an itinerary, which usually does not exceed two or three.

\subsection{Our contributions}
Motivated by the importance of online hypergraph matching for small values of $k$, we focus on $3$-uniform hypergraphs, with the goal of obtaining tighter bounds.
Our main result is a tight competitive ratio for the fractional version of this problem.

\begin{theorem} \label{thm:fractional}
For the online fractional matching problem on $3$-uniform hypergraphs, there is a deterministic $(e-1)/(e+1)$-competitive algorithm.
Furthermore, every algorithm is at most $(e-1)/(e+1)$-competitive.
\end{theorem}

The deterministic algorithm in \Cref{thm:fractional} belongs to the class of {\sc Water-Filling} algorithms.
It uses the function $f(x) \coloneqq e^x/(e+1)$ to decide which hyperedges receive load.
In particular, for every online vertex $w$, the incident hyperedges $h=\{u,v,w\}\in \delta(w)$ which minimize $\phi(h)\coloneqq f(x(\delta(u))) + f(x(\delta(v)))$ receive load until $\phi(h)\geq 1$.

Our main contribution is proving a matching upper bound in \Cref{thm:fractional}.
For this, it suffices to consider deterministic algorithms because every randomized algorithm induces a deterministic fractional algorithm with the same expected value. 
This, in turn, allows us to construct an instance which is adaptive to the actions of the algorithm. 
The key idea is to combine two hard instances for online matching on \emph{bipartite graphs} \cite{karp1990optimal,gamlath2019online}.

We start with the instance in \cite{gamlath2019online}, designed for the edge-arrival model. 
In this instance, edge arrivals are grouped into phases, such that the size of an online maximum matching increases by one per phase.
At the end of every phase, as long as the total fractional value on the revealed edges exceeds a certain threshold, the next phase begins. Otherwise, the instance terminates. 
For our purpose, we want a more fine-grained control over the actions of the algorithm.
So, we apply thresholding at the node level instead, based on fractional degrees, to determine which nodes become incident to the edges arriving in the next phase.

In our construction, we will have multiple copies of this modified edge-arrival instance. 
The edges in these instances are connected to the online nodes to form hyperedges.
The way in which they are connected is inspired by the instance in \cite{karp1990optimal}, originally designed for the vertex-arrival model. 
The idea behind this vertex-arrival instance is to obfuscate the partners of the online nodes in an offline maximum matching, which is also applicable in our setting.

Our next result concerns the online integral matching problem on $k$-uniform hypergraphs.
We show that one can do better than the greedy algorithm if the online nodes have bounded degree.
It is achieved by the simple algorithm {\sc Random}: for every online vertex $w$, uniformly select a hyperedge among all the hyperedges incident to $w$ which are disjoint from the current matching.
\begin{theorem} \label{thm:integral_bounded_degree}
For the online matching problem on $k$-uniform hypergraphs where online vertices have maximum degree $d$, the competitive ratio of {\sc Random} is at least
\[\min\left(\frac{1}{k-1},\, \frac{d}{(d-1)k+1}\right).\] 
\end{theorem}
Note that in \Cref{thm:integral_bounded_degree}, the first term is at most the second term if and only if $d\leq k-1$. Moreover, {\sc Random} is at least as good as the greedy algorithm, since the latter is $1/k$-competitive. For $3$-uniform hypergraphs, the bound becomes $1/2$ for $d \leq 2$ and $1/(3-2/d)$ otherwise, thus interpolating between $1/3$ and $1/2$. Note that for $d \leq 2$, the bound is optimal, since the online matching problem on graphs under edge arrivals is a special case of this setting (with $k = 3, d=1$), for which an upper bound of $1/2$ is known even against fractional algorithms on bipartite graphs \cite{gamlath2019online}.

Since every randomized algorithm for integral matching induces a deterministic algorithm for fractional matching, the upper bound of $(e-1)/(e+1)\approx 0.4621$ in \Cref{thm:fractional} also applies to the integral problem on 3-uniform hypergraphs.
However, the best known lower bound is 1/3, given by the greedy algorithm. An interesting question for future research is whether there exists an integral algorithm better than greedy on 3-uniform hypergraphs.

\subsection{Related work}
Since the online matching problem was introduced in \cite{karp1990optimal}, it has garnered significant interest, leading to extensive follow-up work.
We refer the reader to the excellent survey by Mehta~\cite{mehta_survey} for navigating this rich literature.
The original analysis of {\sc Ranking} \cite{karp1990optimal} was simplified in a series of papers \cite{birnbaum2008line,devanur2013randomized,goel2008online,eden2021economics}.
Many variants of the problem have been studied, such as the online $b$-matching problem \cite{kalyanasundaram2000optimal}, and its extension to the AdWords problem \cite{buchbinder2007online,devanur2012online,huang2020adwords,mehta2007adwords}.
Weighted generalizations have been considered, e.g., vertex weights \cite{aggarwal2011online,huang2019online} and edge weights \cite{fahrbach2022edge}.
Weakening the adversary by requiring that online nodes arrive in a random order has also been of interest \cite{karande2011online,mahdian2011online,kesselheim2013optimal}.
Another line of research explored more general arrival models such as two-sided vertex arrival \cite{conf/icalp/WangW15}, general vertex arrival \cite{gamlath2019online}, edge arrival \cite{journals/algorithmica/BuchbinderST19,gamlath2019online}, and general vertex arrival with departure times \cite{huang_fully_online_1,huang_fully_online_2,ashlagi_online_windowed}.

In contrast, the literature on the online hypergraph matching problem is relatively sparse. Most work has focused on stochastic models, such as the random-order model.
Korula and Pal \cite{conf/icalp/KorulaP09} first studied the edge-weighted version under this model.
For $k$-uniform hypergraphs, they gave an $\Omega(1/k^2)$-competitive algorithm.
This was subsequently improved to $\Omega(1/k)$ by Kesselheim et al.~\cite{kesselheim2013optimal}.
Ma et al.~\cite{journals/ior/MaRST20} gave a $1/k$-competitive algorithm under `nonstationary' arrivals.
Pavone et al.~\cite{journals/ior/PavoneSST22} studied online hypergraph matching with delays under the adversarial model.
At each time step, a vertex arrives, and it will depart after $d$ time steps.
A hyperedge is revealed once all of its vertices have arrived.
Note that their model is incomparable to ours because every vertex has the same delay $d$.

In the prophet IID setting, every online node has a weight function which assigns weights to its incident hyperedges, and these functions are independently sampled from the same distribution.
For this problem, \cite{marinkovic_online_2023} gave a $O(\log (k)/k)$ upper bound on the competitive ratio. 
We refer to \cite{marinkovic_online_2023} for an overview of known results in related settings. 

Hypergraph matching on $k$-uniform hypergraphs is a well-studied problem in  the offline setting. It is known to be impossible to approximate within a factor of $\Omega(1/k)$ in polynomial time, unless $\mathbf{NP}\subseteq \mathbf{BPP}$ \cite{lee_asymptotically_2025}.
Moreover, the factor between the optimal solution and the optimal value of the natural LP relaxation is at least $1/(k-1+1/k)$ \cite{chan_linear_2012}.

A special case that has also been studied is the restriction to $k$-partite graphs, where the vertices are partitioned into $k$ sets and every hyperedge contains exactly one vertex from each set. This setting is called \emph{$k$-dimensional matching}, and the optimal solution is known to be at least $1/(k-1)$ times the optimal value of the standard LP relaxation \cite{chan_linear_2012}. 
For $k=3$, the best known polynomial time approximation algorithm gives a $(3/4-\varepsilon)$-approximation \cite{conf/focs/Cygan13}. 

\subsection{Paper organization}
In \cref{sec_preliminaries}, we give the necessary preliminaries and discuss notation.
\cref{optimal_algo_3_uniform_hypergraphs} presents the optimal primal-dual fractional algorithm for $3$-uniform hypergraphs, which shows the first part of \cref{thm:fractional}.
Section \ref{sec:upper_bound} complements this with a tight upper bound, proving the second part of \cref{thm:fractional}.
The proof of \cref{thm:integral_bounded_degree} is shown in Section \ref{app:bounded_degree_algo}.

\section{Preliminaries}
\label{sec_preliminaries}
Given a hypergraph $\mathcal{H} = (V,H)$ with vertex set $V$ and hyperedge set $H$, the maximum matching problem involves finding a maximum cardinality subset of disjoint hyperedges. The canonical primal and dual LP relaxations for this problem are respectively given by:

\vspace{0.3cm}

\begin{minipage}[t]{0.5\textwidth}
\begin{align*}
    \max \sum_{h \in H} x_h  & \\
    \sum_{h \in \delta(v)} x_h &\leq 1  \quad \forall v \in V\\
    x_h &\geq 0 \quad \forall h\in H
\end{align*}
\end{minipage}
\begin{minipage}[t]{0.5\textwidth}
\begin{align*}
    \min \sum_{v \in V} y_v  & \\
    \sum_{v \in h} y_v &\geq 1 \quad \forall h \in H\\
    y_v &\geq 0 \quad \forall v\in V.
\end{align*}
\end{minipage}

\vspace{0.3cm}

\noindent We denote by $\optLP(\mathcal{H})$ the offline optimal value of these two LPs.
We denote by $\opt(\mathcal{H})$ the objective value of an offline optimal integral solution to the primal LP, which clearly satisfies $\opt(\mathcal{H}) \leq \optLP(\mathcal{H})$.

The online matching problem on $k$-uniform hypergraphs under vertex arrivals is defined as follows.
An instance consists of a $k$-uniform hypergraph $\mathcal{H} = (V, W, H)$, where $V$ is the set of offline nodes and $W=(w_1,w_2,\ldots)$ is the sequence of online nodes. 
The ordering of $W$ corresponds to the arrival order of the online nodes. 
Every hyperedge $h \in H$ has exactly one node in $W$ and $k-1$ nodes in $V$. 
When an online node $w \in W$ arrives, its incident hyperedges $\delta(w)$ are revealed.
A fractional algorithm is allowed to irrevocably increase $x_h$ for every $h \in \delta(w)$, whereas an integral algorithm is allowed to irrevocably pick one of these hyperedges, i.e., setting $x_h = 1$ for some $h\in \delta(w)$.

Given an algorithm $\mathcal{A}$ and an instance $\mathcal{H}$, we denote by 
$\val(\mathcal{A},\mathcal{H}) := \sum_{h \in H} x_h$
the value of the (fractional) matching obtained by $\Alg$ on $\mathcal{H}$. 
An integral algorithm is \emph{$\rho$-competitive} if for any instance $\mathcal{H}$, $\val(\mathcal{A},\mathcal{H}) \geq \rho \; \opt(\mathcal{H})$.
Similarly, a fractional algorithm is \emph{$\rho$-competitive} if for any instance $\mathcal{H}$, $\val(\mathcal{A},\mathcal{H}) \geq \rho \; \optLP(\mathcal{H})$.

In this paper, we focus on 3-uniform hypergraphs.
For a $3$-uniform instance $\mathcal{H} = (V,W,H)$, we denote by $\Gamma(\mathcal{H}) = (V, E)$ the graph on the offline nodes with edge set
\begin{align}
\label{def_E}
E := \Big\{(u,v) \in V \times V, \quad \exists w \in W \text{ s.t. }\{u,v,w\} \in H\Big\}.
\end{align}
We remark that $\Gamma(\mathcal{H})$ is not a multigraph. In particular, an edge $(u,v) \in E$ can have several hyperedges in $H$ containing it. A fractional matching $x$ on the hyperedges $H$ naturally induces a fractional matching $x'$ on the edges $E$, i.e., $x'_e = \sum_{h:e \subseteq h} x_h$ for every $e \in E.$ 
The value obtained by an algorithm $\Alg$ can thus also be counted as $\val(\mathcal{A},\mathcal{H}) = \sum_{h \in H} x_h = \sum_{e \in E} x'_e$. 
For an offline node $u\in V$, we denote its \emph{load} (or \emph{fractional degree}) as $\ell_u = x(\delta(u)) \in [0,1]$.

\section{Optimal fractional algorithm for 3-uniform hypergraphs}
\label{optimal_algo_3_uniform_hypergraphs}

In this section, we present a primal-dual algorithm for the online fractional matching problem on 3-uniform hypergraphs under vertex arrivals.
This algorithm will turn out to be optimal with a tight competitive ratio of $(e-1)/(e+1) \approx 0.4621$. We define the following function $f :[0,1] \to [0,1]$:
\begin{equation}
\label{eq_threshold_funct}
\hfunc(x) := \frac{e^x}{e+1}.
\end{equation}
When an online node $w$ arrives, our algorithm chooses to uniformly increase the primal variables of the hyperedges $\{u,v,w\}$ for which $\hfunc(x(\delta(u)))+\hfunc(x(\delta(v)))$ is minimal. We note that this belongs to the class of \emph{water-filling} algorithms \cite{kalyanasundaram2000optimal}. For this reason, we define the \emph{priority} of a hyperedge $h = \{u,v,w\}$ as:
\begin{equation}
\label{eq_priority_funct}
\phi(h) := \hfunc(x(\delta(u))) + \hfunc(x(\delta(v))).
\end{equation}
Figure \ref{fig_threshold} shows the possible values of $x(\delta(u))$ and $x(\delta(v))$ such that $\phi(h)\leq 1$. We now present the algorithm. The dual updates are shaded in gray since they are only needed for the analysis.

\begin{algorithm}
    \caption{Water-filling fractional algorithm for 3-uniform hypergraphs}
    \label{vertex_arrival_water_algo}
    \begin{algorithmic}
    \State $\mathbf{Input: }$ $3$-uniform hypergraph $\mathcal{H} = (V,W,H)$ with online nodes $W$.
    \State $\mathbf{Output: }$ Fractional matching $x \in [0,1]^H$
    \vspace{0.2cm}
    \State \textbf{when $w \in W$ arrives with $\delta(w) \subseteq H$:}
    \State \quad set $x_h = 0$ for every $h \in \delta(w)$
    \State \quad increase $x_h$ for every
    $h = \{u,v,w\} \in \arg \min_{h \in \delta(w)} \{\phi(h)\}$ at rate $1$
     \State \colorline{\quad increase $y_u$ and $y_v$ at rates $\hfunc(x(\delta(u)))$ and $\hfunc(x(\delta(v)))$}
     \State \colorline{\quad increase $y_w$ at rate $1 - \hfunc(x(\delta(u))) - \hfunc(x(\delta(v)))$ }
     \State \quad \quad \textbf{until} $x(\delta(w)) = 1$ \textbf{or} $\phi(h) \geq 1$ for every $h \in \delta(w)$.
\State \Return $x$
    \end{algorithmic}
    \end{algorithm}
\FloatBarrier

\begin{theorem}
    \cref{vertex_arrival_water_algo} is $(e-1)/(e+1)$-competitive for the online fractional matching problem on $3$-uniform hypergraphs.\label{thm:optimal_algo}
\end{theorem}

\begin{proof}[Proof of \cref{thm:optimal_algo}]
    We first show that the algorithm produces a feasible primal solution. Note that the fractional value of a hyperedge $h$ is only being increased if $\phi(h)\leq 1$. If $x(\delta(v))=1$ for some offline node $v$ then for any hyperedge $h \ni v$, we have:
        \begin{align*}
           \phi(h) = \hfunc(x(\delta(u))) +\hfunc(x(\delta(v)))  \geq \hfunc(1) + \hfunc(0) = \frac{e+1}{e+1}  =1,
        \end{align*}
        where $u$ denotes the second offline node belonging to $h$. The value of the hyperedge $h$ will thus not be increased anymore, proving the feasibility of the primal solution.
        
    In order to prove the desired competitive ratio, we show that the primal-dual solutions constructed during the execution of the algorithm satisfy:
    \begin{align}
    \label{eq_primal_equal_dual_v}
    \val(\mathcal{A}) = \sum_{h \in H} x_h &= \sum_{v \in V \cup W} y_v\qquad \text{and}\\
    \label{eq_bound_comp_ratio_v}
    \sum_{v \in h} y_v &\geq \rho:= \frac{e-1}{e+1} \qquad \forall h \in H.
    \end{align}
    This is enough to imply the desired competitiveness of our algorithm, since $y/\rho \in \mathbb{R}^V_+$ is then a feasible dual solution, giving:
    \[\val(\mathcal{A}) \geq \sum_{v \in V \cup W} y_v \geq \rho \; \optLP .\]
    
        Note that \eqref{eq_primal_equal_dual_v} holds at the start of the algorithm. Let us fix a hyperedge $h = \{u,v,w\} \in H$.
        When $x_h$ is continuously being increased at rate one, the duals on the incident nodes $y_u$, $y_v$ and $y_{w}$
        are being increased at rate $\hfunc(x(\delta(u)))$, $\hfunc(x(\delta(v)))$ and $1-\hfunc(x(\delta(u))) -\hfunc(x(\delta(v)))$ respectively.
        Observe that these rates sum up to one.
        Hence, $\val(\mathcal{A}) = \sum_{h \in H}x_h$ and $\sum_{v\in V \cup W}y_v$ are increased at the same rate, meaning that \eqref{eq_primal_equal_dual_v} holds at all times during the execution of the algorithm.

        We now show that \eqref{eq_bound_comp_ratio_v} holds at the end of the execution of the algorithm. Let us fix an online node $w \in W$. For a given hyperedge $h \in \delta(w)$, note that the algorithm only stops increasing $x_h$, as soon as either $\phi(h) \geq 1$ or $x(\delta(w)) = 1$ is reached. We distinguish these two cases for the analysis.

        Let us first focus on the first case, meaning that $\phi(h) \geq 1$ has been reached for every $h \in \delta(w)$. Consider an arbitrary $h = \{u,v,w\} \in \delta(w)$.
        For every unit of increase in $x(\delta(u))$, $y_u$ will have been increased by $\hfunc(x(\delta(u)))$.
        If we denote by $\ell_u := x(\delta(u))$ and $\ell_v := x(\delta(v))$ the fractional loads on $u$ and $v$ after the last increase on the hyperedges adjacent to $w$, then:
    
        \begin{align}
            y_u=\int_0^{\ell_u}\hfunc(s)ds = \hfunc(\ell_u) - \hfunc(0) \quad \text{and}  \quad y_v=\int_0^{\ell_v}\hfunc(s)ds = \hfunc(\ell_v) - \hfunc(0), \label{eq:dual_val}
        \end{align}
        where we have used the fact that $f$ is an antiderivative of itself. Therefore,
        \begin{align*}
            y_u+y_v+y_w&\geq y_u+y_v= \hfunc(\ell_u)-\hfunc(0)+\hfunc(\ell_v)-\hfunc(0)\\&=\phi(h) -2\hfunc(0)\geq  1 -2\hfunc(0)=\frac{e-1}{e+1}.
        \end{align*}
    
        Suppose now that $x(\delta(w)) = 1$ has been reached. In particular, this means that for each $\{u,v,w\}\in \delta(w)$, the
        rate at which $y_w$ was increased must have been at least $1-\hfunc(\ell_u)-\hfunc(\ell_v)$ at all times, where $\ell_u$ and $\ell_v$ denote the fractional loads on $u$ and $v$ after that the algorithm has finished increasing the edges incident to the online node $w$.
        Hence, we have:
        \begin{align*}
            y_w \geq 1\cdot (1-\hfunc(\ell_u)-\hfunc(\ell_v)).
        \end{align*}
        By using \eqref{eq:dual_val} we see that:
        \begin{align*}
            y_u+y_v+y_w\geq \hfunc(\ell_u) + \hfunc(\ell_v) - 2\hfunc(0) + (1-\hfunc(\ell_u)-\hfunc(\ell_v))= 1 -2\hfunc(0)=\frac{e-1}{e+1}.
        \end{align*}
        This proves \eqref{eq_bound_comp_ratio_v}, and thus completes the proof of the theorem.
        \end{proof}

We note that the above algorithm and its analysis are very similar to those of \cite{buchbinder2009online}. Here we optimized the choice of $f$ for the case of 3-uniform hypergraphs.
By choosing $f(x):= (k\log(k))^{x-1}$ and setting $\phi(h):= \sum_{u\in h} f(x(\delta(u)))$ for every $k$-uniform hyperedge $h$, it can be shown that the natural generalization of \cref{vertex_arrival_water_algo} to $k$-uniform hypergraphs is $\Omega(1/\log(k))$-competitive \cite{journals/orl/TrobstU24}.

\section{Tight upper bound for 3-uniform hypergraphs}
\label{sec:upper_bound}
We now prove the second part of \cref{thm:fractional}, i.e., every algorithm is at most $(e-1)/(e+1)$-competitive for the online fractional matching problem on 3-uniform hypergraphs under vertex arrivals.

\subsection{Overview of the construction}
\label{subsec:construction_overview}
We construct an adversarial instance that is adaptive to the behaviour of the algorithm. The main idea is to combine the vertex-arrival instance of Karp et al.~\cite{karp1990optimal} and the edge-arrival instance of Gamlath et al.~\cite{gamlath2019online} for bipartite graphs. We first briefly describe these two instances, and then give a high-level view of our construction.

\paragraph{Bipartite hardness under vertex-arrivals.}
An illustration of this construction is given in Figure \ref{fig_v_arrival}. This instance consists of a bipartite graph, where one side of the bipartition has size $n \in \mathbb{N}$ and is known upfront. The other side of the bipartition arrives online, with $n$ online nodes arriving in total. The $k$th online node has degree $n-k+1$, with the edges incident to it connecting to the $n-k+1$ offline nodes with the highest fractional degree at that point in time. 

Note that, after all online nodes have arrived, the instance contains a perfect matching of size $n$. It can however be shown that no fractional algorithm can achieve a competitive ratio better than $1- 1/e \approx 0.63$ for this instance, due to the uncertainty about the edge adjacent to each online node in the perfect matching, since it is a decision made online by the adversary. This construction can easily be adapted to our model on $3$-uniform hypergraphs under vertex arrivals, by replacing each offline node by a pair of two identical offline nodes, and replacing each edge of the graph by a hyperedge of size three.

\begin{figure}
\center
\includegraphics[width = 0.95 \textwidth]{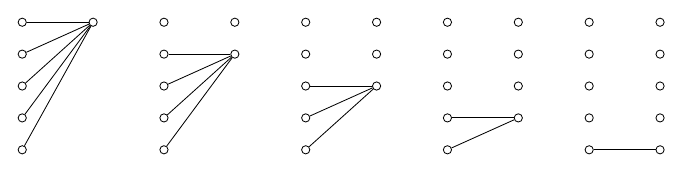}
\caption{The hard instance for bipartite graphs under vertex-arrivals. The $k$th online node connects to the $n - k +1$ offline nodes with the highest fractional degree at that point in time. In this figure, the offline vertices on the left side of the bipartition are drawn in decreasing order of their fractional degrees. The perfect matching consists of all the horizontal edges.}
\label{fig_v_arrival}
\end{figure}

\begin{figure}
\center
\includegraphics[width = 0.95 \textwidth]{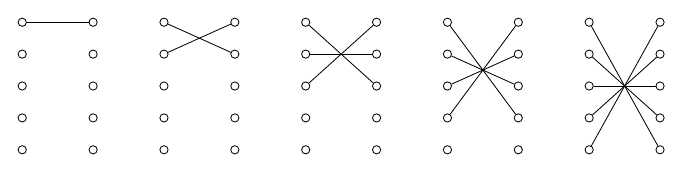}
\caption{The hard instance for bipartite graphs under edge-arrivals. The adversary can terminate the instance at any phase if the algorithm does not generate enough value. The optimal matching consists of all the edges which arrived in the last phase.}
\label{fig_edge_arr_without_threshold}
\end{figure}

\paragraph{Bipartite hardness under edge-arrivals.}
An illustration of this construction is given in Figure \ref{fig_edge_arr_without_threshold}. Each side of the bipartition has cardinality $T$ and the nodes on both sides are denoted by $\{1, \dots, T\}$. This instance now consists of $T' \leq T$ phases, where $T'$ is unknown a priori and determined online by the adversary given the choices of the algorithm. At each phase $t \in [T']$, the following $t$ new edges, which form a matching, are revealed:
\[ \Big\{(i, t -i + 1) : i \in \{1, \dots, t\} \Big\}.\]
Observe that these become the optimal matching at that point in time. If the value generated by the algorithm exceeds $\alpha t$ at the end of phase $t$ for some $\alpha \in [0,1]$, the instance continues to the next phase. Otherwise, the instance terminates. It is shown in \cite{gamlath2019online} that no fractional algorithm can achieve a competitive ratio better than $1/2$ for this instance. The hardness is due to the uncertainty about the time horizon: the adversary can always make a new set of edges arrive if the algorithm gains enough value in each phase. This construction can easily be adapted to our vertex-arrival model on $3$-uniform hypergraphs, by appending an online node with degree $1$ to each edge of the bipartite graph.

\paragraph{Our construction.}
We now give a high-level overview of our construction, which combines both hard instances described above. The offline vertices of the hypergraph are partitioned into $m$ sets $C_1, \dots, C_m$, which we call \emph{components}. Each component will induce a bipartite graph with bipartition $C_i = U_i \cup V_i$, where $U_i$ and $V_i$ are both ordered vertex sets of size $T \in \mathbb{N}$.

The instance consists of $T$ phases. In each phase $t \in \{1, \dots, T\}$, the adversary first selects a bipartite matching $\Match_i^{(t)}$ on each component $C_i$. Taking the union of these matchings gives a larger matching on the offline nodes:
$\Match^{(t)} := \bigcup_{i = 1}^{m}\Match_i ^ {(t)}.$

After selecting the matching $\Match^{(t)}$ at phase $t$, the adversary selects the online nodes, with their incident hyperedges, arriving in that phase. The set of online nodes arriving in phase $t$ is denoted by $W^{(t)}$. Each node $w \in W^{(t)}$ connects to a subset of edges $E(w) \subseteq \Match^{(t)}$, meaning that the hyperedges incident to $w$ are $\{\{w\}\cup e: e\in  E(w)\}$.

We briefly explain how the matchings $\Match_i^{(t)}$ are constructed and how the edges $E(w)$ are picked:
\begin{enumerate}
\item On each component $C_i$, the matching $\Match_i^{(t)}$ is constructed based on the behaviour of the algorithm in phase $t-1$.
It draws inspiration from the edge-arrival instance in \cite{gamlath2019online}, together with the function $\hfunc(x) = e^x/(e+1)$ defined in \eqref{eq_threshold_funct}. The exact construction is described in Section \ref{sec:lastphase} and illustrated in Figure \ref{fig_edge_arr_w_threshold}.
\item For every online node $w \in W^{(t)}$, the edge set $E(w) \subseteq \Match^{(t)}$ is selected based on the behaviour of the algorithm during phase $t$. This part can be seen as incorporating the vertex-arrival instance in \cite{karp1990optimal}. The exact construction is described in Section \ref{sec:firstphases} and illustrated in Figure \ref{fig_vertex_arrival}.
\end{enumerate}
To summarize, the instance is a hypergraph $\mathcal{H} = (V,W,H)$ with offline nodes $V$, online nodes $W$ and hyperedges $H$ given by
\[V \coloneqq \bigcup_{i=1}^m C_i = \bigcup_{i=1}^m U_i \cup V_i \quad\;\, W \coloneqq \bigcup_{t=1}^T W^{(t)} \quad\;\, H \coloneqq \bigcup_{t = 1}^T \bigcup_{w \in W^{(t)}} \{\{w\}\cup e: e\in E(w)\}.\]

\subsection{Assumptions on the algorithm}
\label{subsub:analysis_overview}
To simplify the construction and analysis of our instance, we will make two assumptions on the algorithm.
First, we need the following definition, which relates the behaviour of an algorithm to the priority function $\phi$ defined in \eqref{eq_priority_funct}.

\begin{definition}
\label{def_thresh_resp}
	Fix $\varepsilon\geq 0$. 
	Let $x$ be the fractional solution given by an algorithm $\Alg$ after the arrival of an online node $w$.
	We say that $\Alg$ is \emph{$\varepsilon$-threshold respecting} on $w$ if $\phi(h) = \sum_{v\in h\setminus \{w\}}\hfunc(x(\delta(v)))\leq 1 + \varepsilon$ for all incident hyperedges $h \in \delta(w)$ with $x_h>0$.
    We also call $\Alg$ \emph{threshold respecting} if $\varepsilon = 0$.
\end{definition}

\begin{figure}
\center
\includegraphics[width = 0.4\textwidth]{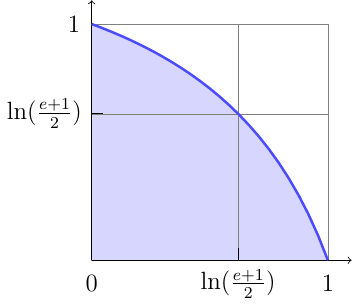}
\caption{An illustration of the region $R = \{(a, b) \in [0,1]^2 : f(a) + f(b) \leq 1\}$. 
The symmetric point at the boundary of the region has both coordinates $\ln((e+1)/2) \approx 0.62$. 
When an online node $w$ arrives, a threshold respecting algorithm ensures that the fractional matching $x$ satisfies $(x(\delta(u)), x(\delta(v))) \in R$ for every hyperedge $h = \{u,v,w\} \in \delta(w)$ with $x_h>0$ at the end of that iteration.}
\label{fig_threshold}
\end{figure}

\begin{remark}We emphasize that the property in \cref{def_thresh_resp} only needs to hold for the fractional solution $x$ after $w$ has arrived, and before the arrival of the next online node. In particular, it is possible that $\phi(h)>1+\varepsilon$ in later iterations. For a hyperedge $h = \{u,v,w\}\in \delta(w)$, Figure \ref{fig_threshold} shows the possible values of $x(\delta(u))$ and $x(\delta(v))$ such that $\phi(h)\leq 1$.
\end{remark}
The two assumptions that we make are the following. In \cref{sec:assumptions}, we show that they can be made without loss of generality. 
\begin{enumerate}
    \item\label{ass:threshold_respecting} The algorithm is \emph{$\varepsilon$-threshold respecting} on all online nodes in the first $T-1$ phases for some arbitrarily small $\varepsilon>0$.
    \item\label{ass:symmetric} The algorithm is \emph{symmetric} on each component $C_i = U_i \cup V_i$. In particular, for every $t\in \{1,\dots,T\}$, the $t$th vertices of $U_i$ and $V_i$ have the same fractional degrees throughout the execution of the algorithm.
\end{enumerate}

\subsection{Constructing the matching $\Match^{(t)}$}
\label{sec:lastphase}
In this section, we construct the matching $\Match_i^{(t)}$ for every component $i\in [m]$ and phase $t\in [T]$. 
We will only describe the matchings for a single component $C_i$, i.e., $\Match_i^{(1)}, \Match_i^{(2)}, \dots, \Match_i^{(T)}$, because the same construction applies to other components.
Intuitively, the value obtained by the algorithm in the first $T-1$ phases is already limited by Assumption~\ref{ass:threshold_respecting}. 
So, the goal of this construction is to prevent the algorithm from gaining too much value in the last phase $T$. 
In particular, we will show that it can only obtain $O(\sqrt{T}) + \varepsilon \; O(T^2)$ in the last phase on every component $C_i$. 

The matchings $\Match_i^{(1)}, \dots, \Match_i^{(T)}$ are adaptive to the behaviour of the algorithm in every phase. 
It is essentially the instance of \cite{gamlath2019online} with our threshold function incorporated. The vertex set of these matchings is on a bipartite graph, with $T$ nodes on both sides of the bipartition. Let us denote this bipartition as $U_i = \left\{1, \dots, T\right\}$ and $V_i = \left\{1, \dots, T\right\}$. We index them the same way due to the symmetry assumption of the algorithm (Assumption~\ref{ass:symmetric}). 
Each matching $\Match_i^{(t)}$ satisfies the invariant that $(u,v)\in \Match_i^{(t)}$ if and only if $(v,u)\in \Match_i^{(t)}$.

For an offline node $u\in V$, we denote its \emph{load} (or \emph{fractional degree}) at the end of phase $t$ as $\ell^{(t)}_u = x^{(t)}(\delta(u)) \in [0,1]$, where $x^{(t)}$ is the fractional matching generated by the algorithm at the end of phase $t$.

\begin{itemize}
\item $\Match_i^{(1)}$ is a matching of size one that consists of the single edge $(1, 1)$. 

\item At the end of phase $t$, we will call a node \emph{active} if it is incident to an edge $e = (u,v) \in \Match_i^{(t)}$ satisfying $\phi(e) = \hfunc(\ell_u^{(t)}) +  \hfunc(\ell_v^{(t)}) \geq 1$. All other nodes are said to be \emph{inactive} and will not be used in any of the matchings of later phases.
Let $\sigma_{t}(1)<\sigma_{t}(2)<\ldots <\sigma_{t}(\numact_{t})$ be the active nodes in $U_i$ 
at the end of phase $t$, where $r_{t}$ denotes the number of such active nodes. 
By the aforementioned invariant and Assumption~\ref{ass:symmetric}, the active nodes in $V_i$ are also $\sigma_{t}(1)<\sigma_{t}(2)<\ldots <\sigma_{t}(\numact_{t})$.  \item The matching at phase $t+1$ is then of size $\numact_{t} +1$ and is defined as:
\[\Match_i^{(t+1)} := \Big\{\Big(\sigma_t(k), \sigma_t(r_t + 2 - k)\Big), k \in \{1, \dots, r_t+1\}\Big\},\]
where we define $\sigma_t(\numact_t+1) := t+1$ for convenience.
In particular, note that $t+1 \in U_i$ and $t+1 \in V_i$ are two fresh nodes with zero load, which are always part of the matching $\Match_i^{(t+1)}$, but not part of any matching from a previous phase. Clearly, the invariant is maintained. 
Figure \ref{fig_edge_arr_w_threshold} illustrates the construction.
\end{itemize}

\begin{figure}
\center
\includegraphics[width = 0.95\textwidth]{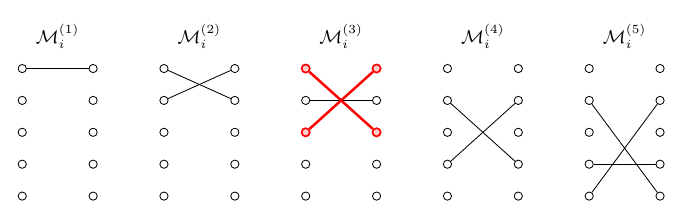}
\caption{In this example, the algorithm does not increase the edge $(1,3)$, and thus by symmetry the edge $(3,1)$, up to the threshold during phase $t = 3$. Hence, $\hfunc(\ell_1^{(3)}) + \hfunc(\ell_3^{(3)}) < 1$ and  nodes $1$ and $3$ become inactive from that point on. The maximum matching at the end of phase $5$ still has size five and consists of $\Match_i^{(5)}$, in addition to the two edges $(1,3)$ and $(3,1)$ that are below the threshold.}
\label{fig_edge_arr_w_threshold}
\end{figure}

Let us denote $q_t := (r_t+1)/2$. Observe that the nodes $\sigma_t(k) \in U_i$ and $\sigma_t(k) \in V_i$ for every $k \in \{1, \dots, \lceil q_t\rceil \}$ form a vertex cover of the matching $\Match_{i}^{(t+1)}$, meaning that every edge of the matching in phase $t+1$ is covered by one of these active nodes at phase $t$. Intuitively, this construction ensures that as $t$ gets large, these nodes have a high fractional degree. Consequently, the algorithm does not have a lot of room to increase the fractional value on any edge of $\Match_{i}^{(t+1)}$, due to the degree constraints. In order to upper bound the value that the algorithm can get in phase $t+1$, we will thus lower bound the fractional degree of the active nodes $\sigma_t(i)$ for $i \in \{1, \dots, \lceil q_t \rceil\}$. For this reason, we define:
\[\ell(t,i) := x^{(t)}\Big(\delta(\sigma_t(i))\Big) = \sum_{e \in \delta(\sigma_t(i))} x^{(t)}_e.\]
In words, this is the fractional degree of the $i^{th}$ active node at the end of phase $t$. One can now see $\{\ell(t,i)\}_{t,i}$ as a process with two parameters, which depends on the behaviour of the algorithm. To analyze this process, we will relate it to the CDF of the binomial distribution $B(t,1/2)$. We will in fact show that
\begin{equation}
\label{eq_loads}
\sum_{i = 1}^{\lceil q_{T-1} \rceil} 2(1 - \ell(T-1,i)) = O (\sqrt{T}) + \varepsilon \: O(T^2).
\end{equation}
Since the left-hand side is the residual capacity of the vertex cover of $\mathcal{M}^{(T)}_i$, this will yield an upper bound on the value obtained by the algorithm in component $C_i$ during the last phase. For intuition, Figure \ref{fig_process_algo} provides an example of $\ell(t,i)$ if the algorithm exactly reaches the threshold for every edge.

\begin{figure}[t]
	\center
	\begin{subfigure}{0.45 \textwidth}
		\center
		\includegraphics[width = \textwidth]{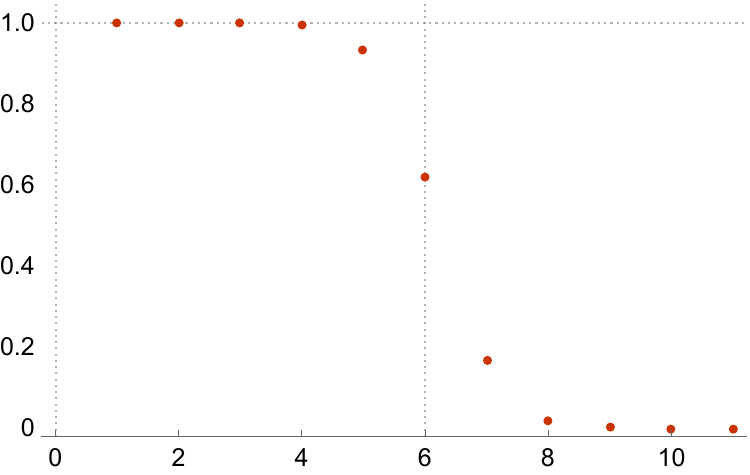}
		\caption{$t = 11, q_t = 6$}
		\label{figure_hsk_a}
	\end{subfigure}
	\hspace*{1cm}
	\begin{subfigure}{0.45\textwidth}
		\center
		\includegraphics[width = \textwidth]{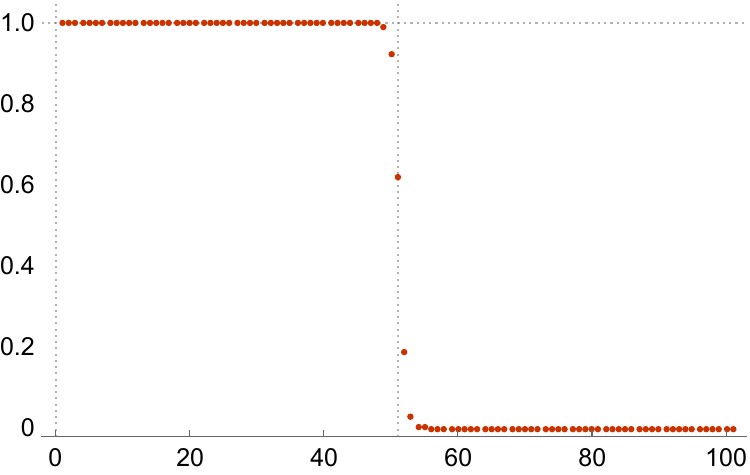}
		\caption{$t = 101, q_t = 51$}
	\end{subfigure}
	\caption{Plot of the $\ell(t,i)$ process for two different values of $t$ if the algorithm exactly matches the threshold at every phase. Observe that, since $t$ is odd, $(\sigma_t(q_t), \sigma_t(q_t)) \in \Match_{i}^{(t)}$ with the load of node $\sigma_t(q_t)$ staying at $\ln((e+1)/2) \approx 0.62$.}
	\label{fig_process_algo}
\end{figure}

\subsection{Bound for the last phase $T$}
In this section, we prove the following theorem by showing \eqref{eq_loads}.
\begin{theorem}
\label{thm_bound_last_phase}
During the last phase $T$, the value gained by the algorithm in each component $C_i$ is at most $O(\sqrt{T}) + \varepsilon \; O(T^2)$.
\end{theorem}
In order to be able to get a lower bound on $\ell(t,i)$, we now relate it to a process which is simpler to analyze, defined as follows on $\mathbb{N} \times \mathbb{Z}/2$:
\begin{align*}
    \psi(t,y)=\Pr_{X\sim B(t,\frac12)}\left[X< \frac{t}{2}+y\right] + \frac12 \Pr_{X\sim B(t,\frac12)}\left[X= \frac{t}{2}+y\right],
\end{align*}
where $B(t, \frac12)$ is the binomial distribution with parameters $t$ and $\frac12$ (see Figure \ref{fig_process_binomial} for an illustration). This process has the following crucial properties, whose proofs are given in \cref{app:properties_psi}.

\begin{figure}[t]
   \center
   \begin{subfigure}{0.4 \textwidth}
       \center
       \includegraphics[width = \textwidth]{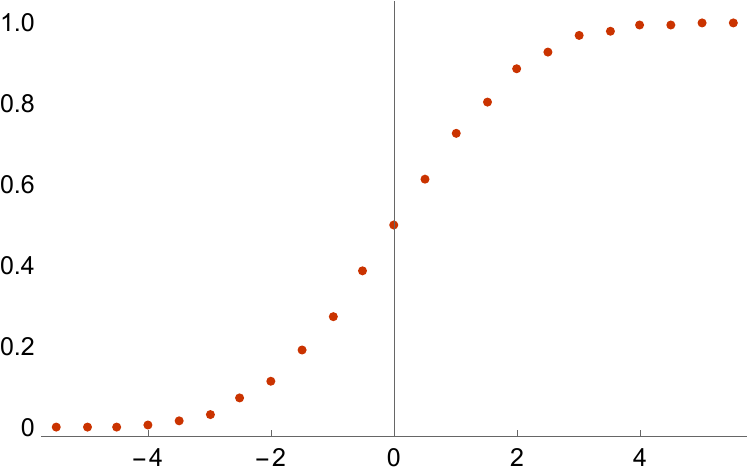}
       \caption{$t = 11$}
       \label{figure_hsk_a}
   \end{subfigure}
   \hspace*{2cm}
   \begin{subfigure}{0.4\textwidth}
       \center
       \includegraphics[width = \textwidth]{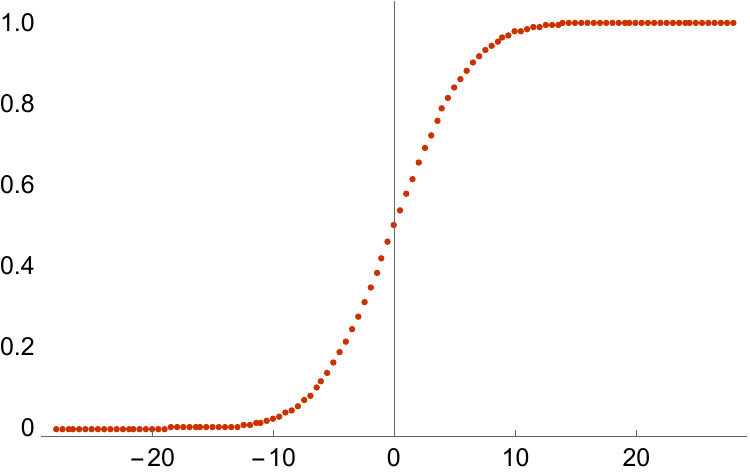}
       \caption{$t = 101$}
   \end{subfigure}
   \caption{Plot of $\psi(t,y)$ for two different values of $t$, the horizontal axis represents $y \in \mathbb{Z}/2$. Here $\epsilon=0$, as it should be thought of as a very small constant.}
   \label{fig_process_binomial}
\end{figure} 

\begin{lemma}
\label{lemma_psi}
    The function $\psi$ satisfies:
    \begin{enumerate}
        \item For all $t$, $\psi(t,y)$ is nondecreasing in $y$ with $\psi(t,0) = \frac12$ and  $\psi(t,\frac{t+1}{2})= 1$.
        \item For all $t$ and $y$, we have: $\psi(t+1,y) = \frac12\psi(t,y-\frac12) + \frac12\psi(t,y+\frac12)$.
        \item For all $t$, we have $\sum_{y=0}^\infty (1-\psi(t,\frac12 y))\leq 1+\sqrt{t}$.
    \end{enumerate}
\end{lemma}

We then relate the process $\ell(t,i)$ to a linear transformation of the process $\psi(t,i)$, by defining:
\[\xi(t,i) := a\; \psi\left(t,q_t - i\right)+b - \varepsilon t.\]
We choose $a$ and $b$ such that whenever the algorithm exactly hits the threshold for every edge, we have $\xi(t,t/2)+\epsilon t=\ln((e+1)/2)=\ell(t,t/2)$ for all even $t$ and $\lim_{t\to\infty}\xi(t,x) + \epsilon t =1=\lim_{t\to\infty}\ell(t,x)$ for all $x$. Hence, we set $a := 2 - 2\ln\Big((e+1)/2)\Big) \approx 0.76$ and $b := 2 \ln\Big((e+1)/2\Big) - 1 \approx 0.24$. This choice will allow us to lower bound $\ell(t,i)$ by $\xi(t,i)$ for all $i\in \{1, \dots, \lceil q_t \rceil\}$ in \cref{lemma_lower_bound_levels}.

We first need one more lemma. As a reminder, $r_t$ is the number of active nodes on each side of the bipartition at the end of phase $t$ and $q_t := (r_t+1)/2$. For every active node $u = \sigma_t(i)$, where $i \leq r_t$, let us define $d_t(u)=|i-q_t|$, the distance of $u$ to the middle point $q_t$ of the active nodes at time $t$.
\begin{lemma}
    \label{lem:distance}
    Let $u=\sigma_t(i)$ for $i\leq r_t$, then the following holds:
    \begin{align*}
        d_t(u)\leq d_{t-1}(u) + \frac12 \quad \text{if } i < q_t \qquad  \text{ and } \qquad d_t(u)\leq d_{t-1}(u) - \frac12 \quad \text{if } i > q_t.
    \end{align*}
\end{lemma}
\begin{proof}
    Let $S$ be the set of indices $i\leq r_{t-1} + 1$ such that $\sigma_{t-1}(i)$ is not active after phase $t$.
    Let $j$ be such that $\sigma_{t-1}(j)=u$. Consider the case that $i<q_t$. Let $c=|s\in S: i<s<r_{t-1}+2-i |$.
    We have $d_t(u)=q_t-i=q_{t-1} +\frac12-\frac12c -i\leq q_{t-1} +\frac12 -i= d_{t-1}(u)+\frac12$. The proof for $i>q_t$ is similar.
\end{proof}

We are now ready to prove the desired lower bound on the loads.

\begin{lemma}
\label{lemma_lower_bound_levels}
For every $t\in \{1, \dots, T-1\}$ and $i\in \{1, \dots, \lceil q_t \rceil\}$, we have \[\ell({t,i})\geq \xi(t,i),\]
where $q_t = (r_t + 1)/2$ and $r_t$ is the number of active nodes at the end of phase $t$, when $r_t\geq 1$.
\end{lemma}

\begin{proof}
We will prove this statement by induction on $t$. \medskip

\textbf{Base case.} For the base case, consider $t = 1$. There are two possibilities, either the edge $(1,1)$ does not make it to the threshold, i.e. $2\hfunc(\ell_1^{(1)}) < 1$, in which case $r_t = 0$, $q_t = 0$, and the statement is then trivially satisfied. If the edge $(1,1)$ makes it to the threshold, then $2 \hfunc(\ell_1^{(1)}) = 2 \hfunc(\ell(1,1)) \geq 1$, which is equivalent to $\ell(1,1) \geq \ln((e+1)/2)$ by definition of $\hfunc(x) = e^x/(e+1)$. Observe that in this case $r_t = q_t = 1$, leading to
\[\ell(1,1) \geq \ln((e+1)/2) = \frac{a}{2} + b = a \: \psi(1,0) + b = \xi(1,1) + \varepsilon t \geq \xi(1,1), \]
where we have used the fact that $\psi(1,0) = 1/2$. 
\medskip

 \textbf{Inductive step.}
Suppose now by induction that the statement holds for $t-1$, let $r_t$ be the number of active nodes at the end of phase $t$, let $q_t := (r_t+1)/2$ and consider an arbitrary $i \in \{1, \dots, \lceil q_t \rceil\}$.
We consider the cases where $i = q_t$ and $i < q_t$ separately. \medskip

Let us first consider the case where $i = q_t = (r_t + 1)/2$, which can only occur when $r_t$ is odd. Observe that this means the edge $(\sigma_t(i), \sigma_t(i))$ belongs to the matching $\Match_i^{(t)}$ and exceeds the threshold, i.e. $\ell(t,i) \geq \ln((e+1)/2)$. Using the fact that $\psi(t,0) = 1/2$ for all $t$, $i = q_t$ and the exact same arguments as above, we get
\[\ell(t,i) \geq \ln((e+1)/2) = \frac{a}{2} + b = a \: \psi(t,0) + b = a \: \psi(t,q_t - i) + b = \xi(t,i) + \varepsilon t \geq \xi(t,i).\]\medskip

Consider now the case where $i < q_t$. Let $e = (u,v) = (\sigma_t(i), \sigma_t(r_t + 1 - i)) \in \Match_i^{(t)}$ and observe that $e$ exceeds the threshold, i.e. $\hfunc(\ell^{(t)}_u) + \hfunc(\ell^{(t)}_v) = \hfunc(\ell(t,i)) + \hfunc(\ell(t, r_t + 1 -i)) \geq 1$.
Let us pick indices $j,k$ such that $u = \sigma_{t-1}(j)$ and $v = \sigma_{t-1}(r_{t-1} + 1 - k)$.
We have:
\begin{align*}
\ell(t,i) = \ell(t-1, j) + x_e^{(t)} \quad \text{and} \quad \ell(t,r_t+1-i) = \ell(t-1, r_{t-1} + 1 -k) + x_e^{(t)}.
\end{align*}

If $k=0$, then $v$ has not appeared in any of the prior matchings, so $\ell^{(t-1)}_v = 0$.
In particular, we have $\hfunc(\ell_v^{(t-1)})=\hfunc(0)= 1-\hfunc(1)\leq 1+\varepsilon - \hfunc(\xi(t-1, k ))$.

Otherwise, we have $(\sigma_{t-1}(k), v) \in \Match_{i}^{(t-1)}$ and by using the fact that the algorithm is $\varepsilon$-threshold respecting, we get $\hfunc(\ell({t-1,k}))+\hfunc(\ell_v^{(t-1)})\leq 1+\varepsilon$. By using the inductive hypothesis $\ell({t-1,k}) \geq \xi(t-1,k)$, we get
$\hfunc(\ell_v^{(t-1)})\leq 1+\varepsilon - \hfunc(\xi(t-1,k)).$
Since edge $e$ exceeds the threshold at the end of phase $t$, we have $\hfunc(\ell^{(t-1)}_u + x_e^{(t)}) + \hfunc(\ell^{(t-1)}_v + x_e^{(t)}) \geq 1$, which leads to
\begin{align*}
x_e^{(t)} &\geq -\ln \Big(\hfunc(\ell_u^{(t-1)})+\hfunc(\ell_v^{(t-1)})\Big) \\
&\geq -\ln\Big(1+\varepsilon - \hfunc(\xi(t-1,k)) +\hfunc(\ell_u^{(t-1)})\Big) \tag*{(by induction hypothesis)}\\
&\geq  \hfunc(\xi(t-1,k)) -\hfunc(\ell_u^{(t-1)}) - \varepsilon \qquad \tag*{(since $\ln(1+x)\geq x$)} \\
&\geq f'(\ell_u^{(t-1)})\Big(\xi(t-1,k)-\ell_u^{(t-1)}\Big)  - \varepsilon \qquad \tag*{\Big(by convexity of $f(x)=\frac{\exp(x)}{e+1}$\Big)}\\&= f'(\ell_u^{(t-1)})\Big(\xi(t-1,k)-\ell(t-1, j)\Big)  - \varepsilon.
\end{align*}
Finally, since we have $f'(\ell_u^{(t-1)})=\hfunc(\ell_u^{(t-1)})\geq  f\Big(\ln((e+1)/2)\Big)= 1/2$, so that $x_e^{(t)}\geq \frac12 \Big(\xi(t-1,k)-\ell(t-1,j)\Big)  - \varepsilon$, and hence:
    \begin{align*}
        \ell(t,i)&= \ell_u^{(t-1)} + x_e^{(t)} \geq \ell_u^{(t-1)} + \frac12 \Big(\xi(t-1,k)-\ell(t-1,j)\Big)  - \varepsilon\\
        &=\frac12 \ell(t-1,j) + \frac12\xi(t-1,k)  - \varepsilon \\&\geq \frac12\xi(t-1,j) + \frac12\xi(t-1,k)  - \varepsilon\quad \tag*{(by induction hypothesis)}\\
        &=\frac{a}{2}\Big(\psi\Big(t-1,d_{t-1}(u)\Big) + \psi\Big(t-1,d_{t-1}(v)\Big)\Big) +b - \varepsilon t\\
        &\geq \frac{a}{2}\Big(\psi\Big(t-1,d_{t}(u)+ \frac{1}{2}\Big) +
        \psi\Big(t-1,d_{t}(v)- \frac{1}{2}\Big)\Big) +b - \varepsilon t \quad \tag*{(by \cref{lem:distance})}\\
        &= a \; \psi\left(t,q_t -i\right) +b - \varepsilon t  \quad \tag*{(by property (2) in \cref{lemma_psi})} \\
        &= \xi(t,i)\\
    \end{align*}
This concludes the proof of the inductive step and thus of the lemma.
\end{proof}
We are now ready to bound the value obtained by the algorithm in the last phase and thus prove Theorem \ref{thm_bound_last_phase}.
\begin{proof}[Proof of \cref{thm_bound_last_phase}]
Consider the end of phase $T-1$. Observe that the nodes $\sigma_{T-1}(k) \in U_i$ and $\sigma_{T-1}(k) \in V_i$ for $k \in \{1, \dots, \lceil q_{T-1}\rceil \}$ form a vertex cover of the final matching $\Match_{i}^{(T)}$. Because of the degree constraints, this means that the value the algorithm can gain on the last phase $T$ is at most twice the following expression:
\begin{align*}
&\sum_{k = 1}^{\lceil q_{T-1} \rceil} \Big(1 - \ell(T-1,k)\Big)\leq  \sum_{k = 1}^{\lceil q_{T-1} \rceil} \Big(1 - \xi(T-1,k)\Big)\\
&=  a\sum_{k = 1}^{\lceil q_{T-1} \rceil} \Big(1 -  \; \psi(T-1,q_{T-1} - k) \Big)+ \varepsilon (T-1)\lceil q_{T-1} \rceil \quad \tag*{(definition of $\xi$ and $1-b=a$)}\\
& \leq \varepsilon \: T^2 + a\sum_{y = -1}^{\infty} \left(1 -  \: \psi\left(T-1, \frac{y}{2}\right)\right) \quad \tag*{(by change of variables $y := 2(q_{T-1} - k)$)}\\&\leq  \varepsilon \: T^2 + a \left(2+\sqrt{T-1}\right) \qquad \tag*{(by property (3) of \cref{lemma_psi})} \\
&= O(\sqrt{T}) + \varepsilon \; O(T^2).
\end{align*}
\end{proof}

\subsection{Connecting the matching $\mathcal{M}^{(t)}$ to the online nodes}
\label{sec:firstphases}

In this section, we connect the matching $\Match^{(t)} = \cup_{i=1}^m \Match_i^{(t)}$ to the online vertices to form hyperedges, for every phase $t\in [T]$.
The way in which they are connected is similar to the vertex-arrival instance of \cite{karp1990optimal} for bipartite graphs.
The main idea is to obfuscate the partners of the online nodes in the optimal matching.

The following construction is an adaptation of the vertex-arrival instance for bipartite graphs \cite{karp1990optimal} to tripartite hypergraphs. Given a graph matching $\Match$ on the offline nodes, the first online node connects to every edge in $\Match$. After the algorithm sets fractional values on every edge of $\Match$, the second online node connects to $\Match \setminus \{e_1\}$, where $e_1$ is the edge in the matching with the lowest fractional value. More generally, for every $k \in \{1, \dots, |\Match|\}$, the $k^{th}$ online node connects to the $|\Match|-k+1$ edges $\Match \setminus \{e_1, \dots, e_{k-1}\}$, and $e_k$ is defined as the edge having the lowest fractional value among them at the end of the $k^{th}$ iteration. This instance is illustrated in Figure \ref{fig_vertex_arrival}.

We now state the guarantee obtained by this construction, parametrized by the maximum (fractional) degree $\Delta \in [0,1]$ attained by an offline node. The proof of the following lemma is shown in \cref{app:proofs_upper_bound}.
\begin{lemma}
\label{lemma_vertex_arrival}
For any graph matching $\Match = (V, E)$, there exists an online tripartite hypergraph instance $\mathcal{H} = (V, W, H)$ such that $\Gamma(\mathcal{H}) = \Match$ and $\opt(\mathcal{H}) = |\Match|$.
Moreover, for any fractional algorithm $\Alg$ whose returned solution $x$ satisfies $x(\delta(v)) \leq \Delta$ for all offline nodes $v\in V$, we have
\[\val(\Alg, \mathcal{H}) \leq (1 - e^{-\Delta})|\Match| + 3/2.\]
\end{lemma}

\begin{figure}
\center
\includegraphics[width = 0.85\textwidth]{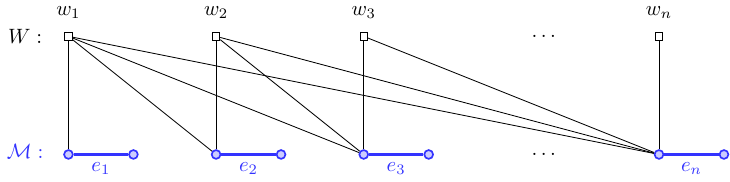}
\caption{An illustration of the instance constructed in Lemma \ref{lemma_vertex_arrival}}
\label{fig_vertex_arrival}
\end{figure}

The way now in which we apply this construction is by partitioning the matching $\Match^{(t)}$ into submatchings based on the load (or fractional degree) of the vertices, and applying \Cref{lemma_vertex_arrival} on each submatching separately. More precisely, let us fix $\mesh := |\mathcal{M}^{(t)}|^{-1/3}$ and $N:= \left\lceil 2 / \mesh \right\rceil$. 

We partition the edges of the matching $\Match^{(t)}$ into $N^2$ submatchings as follows
\[\Match^{(t)}(i,j) := \left\{(u,v) \in \Match^{(t)} : \ell_u \in \left[\frac{i-1}{N}, \frac{i}{N}\right] \; , \; \ell_v \in \left[\frac{j-1}{N}, \frac{j}{N}\right] \right\}\]
for all $i,j \in [N]$. 
Then, we apply the construction illustrated in Figure \ref{fig_vertex_arrival} to each submatching $\Match^{(t)}(i,j)$. 
This finishes the description of our instance. 

\subsection{Bound for the first $T-1$ phases}

Recall from \Cref{subsec:construction_overview} that our instance consists of $T$ phases. 
For ease of analysis, we will split the total value gained by the algorithm into the value gained in each phase. 
For an algorithm $\mathcal{A}$, let $\val^{(t)}(\mathcal{A})$ denote the value obtained by $\mathcal{A}$ in phase $t$.
The next lemma upper bounds $\val^{(t)}(\mathcal{A})$ for a threshold-respecting algorithm $\mathcal{A}$, in terms of the loads of the offline nodes at the end of phase $t-1$.
Recall that $W^{(t)}$ is the set of online nodes which arrive during phase $t$.

\begin{lemma}
\label{lemma_vertex_arrival_blocks}
If $\mathcal{A}$ is threshold-respecting on $W^{(t)}$, then 
\[\val^{(t)}(\mathcal{A}) \leq \sum_{(u,v) \in \Match^{(t)}} \Big(1 - \hfunc(\ell^{(t-1)}_u) - \hfunc(\ell^{(t-1)}_v)\Big) + 15\, |\Match^{(t)}|^{2/3}.\]
\end{lemma}

\begin{proof} 
Since we are considering a threshold-respecting algorithm, we can compute an upper bound on the amount that the algorithm can put on an edge $e \in \Match^{(t)}(i,j)$ while staying below the threshold:
\begin{align*}
\Delta(i,j)&:=\max\left\{x:\hfunc\left(\frac{i-1}{N}+x\right)+\hfunc\left(\frac{j-1}{N}+x\right) \leq 1 \right\} 
= \ln \left(\frac{e+1}{\exp\left(\frac{i-1}{N}\right)+\exp\left(\frac{j-1}{N}\right)}\right).
\end{align*}
A simpler way to write this equation is as follows:
\[\exp(-\Delta(i,j)) = \hfunc\left(\frac{i-1}{N}\right) + \hfunc\left(\frac{j-1}{N}\right).\]
By Lemma \ref{lemma_vertex_arrival}, we know that there exists sets of online nodes $W^{(t)}(i,j)$ which, together with the matchings $\Match^{(t)}(i,j)$, form online hypergraphs $\mathcal{H}^{(t)}(i,j)$ such that
\begin{align*}
\val\Big(\mathcal{A}, \mathcal{H}^{(t)}(i,j)\Big) &\leq \Big(1 - \exp(-\Delta(i,j)\Big) \: \Big|\Match^{(t)}(i,j)\Big| + \frac{3}{2} \qquad \forall i,j \in [N].
\end{align*}
Now, observe that for an edge $\{u,v\} \in \Match^{(t)}(i,j)$ with loads $\ell^{(t-1)}_u$ and $\ell^{(t-1)}_v$, we have
\[\hfunc\left(\frac{i-1}{N}\right) \geq \hfunc\left(\ell^{(t-1)}_u - \frac{1}{N}\right) \geq \hfunc(\ell^{(t-1)}_u) -\frac{1}{N}.\]
The first inequality follows from the fact that $\hfunc$ is a non-decreasing function. The second inequality follows from the fact that $\hfunc'(x) = \hfunc(x) \leq 1$ for every $x \in [0, 1]$. 
Similarly, we have
\[\hfunc\left(\frac{j-1}{N}\right) \geq \hfunc(\ell^{(t-1)}_v) - \frac{1}{N}.\]
Hence, the value gained by the algorithm in phase $t$ can be upper bounded as
\begin{align*}
\val^{(t)}(\mathcal{A}) &= \sum_{i,j=1}^{N}\val\Big(\mathcal{A}, \mathcal{H}^{(t)}(i,j)\Big) \\
&\leq \sum_{i,j=1}^{N}\left(1 - \hfunc\left(\frac{i-1}{N}\right) - \hfunc\left(\frac{j-1}{N}\right)\right)\Big|\Match^{(t)}(i,j)\Big| + \frac{3}{2}N^2 \\
&\leq \sum_{i,j=1}^{N} \sum_{(u,v)\in \Match^{(t)}(i,j)}\left(1 - \hfunc(\ell^{(t-1)}_u) - \hfunc(\ell^{(t-1)}_v) + \frac{2}{N}\right) + \frac{3}{2}N^2 \\
&= \sum_{(u,v) \in \Match^{(t)}} \left(1 - \hfunc(\ell^{(t-1)}_u) - \hfunc(\ell^{(t-1)}_v)\right) + \frac2N |\Match^{(t)}| + \frac{3}{2}N^2  \\
&\leq \sum_{(u,v) \in \Match^{(t)}} \left(1 - \hfunc(\ell^{(t-1)}_u) - \hfunc(\ell^{(t-1)}_v)\right) + \mesh |\Match^{(t)}| + 14\mesh ^ {-2}.
\end{align*} 
For the last inequality, since $N = \lceil 2/\mesh \rceil$, we have used the bounds $2/N \leq \mesh$ and $N^2 \leq (2/\mesh + 1)^2 = (4 + 4\mesh + \mesh^2)/\mesh^2 \leq 9/\mesh^2$ because $\mesh \in (0,1]$.
In fact, $\mesh = |\Match^{(t)}|^{-1/3}$ so the bound becomes
\[\val^{(t)}(\mathcal{A}) \leq \sum_{(u,v) \in \Match^{(t)}} \left(1 - \hfunc(\ell^{(t-1)}_u) - \hfunc(\ell^{(t-1)}_v)\right) + 15|\Match^{(t)}|^{2/3}.\]
\end{proof}

From the definition of $f$ and the construction of the matching $\mathcal{M}^{(t)}$, we can convert the previous bound into the following expression. We remark that the threshold-respecting property is only used in the proof of \Cref{lemma_vertex_arrival_blocks}.
\begin{lemma} 
\label{lem:value_per_phase}
If $\mathcal{A}$ is threshold-respecting on $W^{(t)}$, then 
\[\val^{(t)}(\mathcal{A}) \leq \frac{e-1}{e+1} \,  m + 15\, t^{2/3} m^{2/3}.\]
\end{lemma}

\begin{proof}
By splitting the matching $\Match^{(t)}$ based on the $m$ components, we can rewrite the bound in \Cref{lemma_vertex_arrival_blocks} as
\begin{align}
\label{eq_bound_induction}
\val^{(t)}(\mathcal{A}) &\leq \sum_{i=1}^m \sum_{(u,v) \in \Match^{(t)}_i} \left(1 - \hfunc(\ell^{(t-1)}_u) - \hfunc(\ell^{(t-1)}_v)\right) + 15\,|\Match^{(t)}|^{2/3}.
\end{align}
Fix a component $i \in [m]$, and let $\mathcal{E}_i:= \Big\{e \in \Match^{(t-1)}_i \mid \hfunc(\ell_u^{(t-1)}) + \hfunc(\ell_v^{(t-1)}) \geq 1\Big\}$ be the subset of edges in the matching $\Match^{(t-1)}_i$ which exceed the threshold at the end of phase $t-1$. 
By the construction of $\Match^{(t)}_i$ in \cref{sec:lastphase}, we know that its node set consists of the nodes incident to $\mathcal{E}_i$, in addition to two new fresh nodes whose load is 0 at the end of phase $t-1$. This allows us to expand the inner sum in \eqref{eq_bound_induction} as:
\begin{align*}
\sum_{(u,v) \in \Match^{(t)}_i} 1 - \hfunc(\ell^{(t-1)}_u) - \hfunc(\ell^{(t-1)}_v) &= 1 - 2 \hfunc(0) + \sum_{(u,v) \in \mathcal{E}_i} 1 - \hfunc(\ell^{(t-1)}_u) - \hfunc(\ell^{(t-1)}_v)\\
&\leq 1 - 2 \hfunc(0) = \frac{e-1}{e+1}
\end{align*}
where the inequality follows from definition of $\mathcal{E}_i$.
Plugging this into \eqref{eq_bound_induction} with the bound $|\Match^{(t)}| \leq tm$ (which is immediate by the construction in \cref{sec:lastphase}) yields the desired result.
\end{proof}

For an $\varepsilon$-threshold-respecting algorithm, we pick up an extra $\varepsilon t m$ term.

\begin{corollary} 
\label{cor_total_value_first_phases}
If $\mathcal{A}$ is $\varepsilon$-threshold-respecting on $W^{(t)}$ for some $\varepsilon\geq 0$, then 
\[\mathcal{V}^{(t)}(A) \leq \frac{e-1}{e+1}m + 15 \: (tm)^{2/3} + \varepsilon t m.\]
\end{corollary}
\begin{proof} 
Fix an edge $e\in \Match^{(t)}$.
Let $h_1, h_2, \dots, h_k$ be the hyperedges arriving in phase $t$ which contain $e$, denoted such that $h_i$ arrives before $h_j$ if and only if $i<j$.
Let $j\in [k]$ be the smallest index such that $\phi(h_j)> 1$ immediately after $\Alg$ assigns $x_{h_j}$ to $h_j$.
Let $z_{h_j}\geq 0$ be the largest value such that $\phi(h_j)\geq 1$ if $\Alg$ were to assign $x_{h_j}-z_{h_j}$ to $h_j$ instead.
Define $z_{h_i}:= 0$ for all $i<j$, and $z_{h_i}:= x_{h_i}$ for all $i>j$.
Since $\Alg$ is $\varepsilon$-threshold-respecting on $W^{(t)}$, we have $\sum_{i=1}^k z_{h_i} \leq \varepsilon$ because $f$ is convex and $f' = f$. 

Let $z$ be the vector obtained by repeating this procedure on every edge $e\in \Match^{(t)}$.
Then, $\1^\top z \leq \varepsilon tm$ as $|\Match^{(t)}| \leq tm$.
Moreover, observe that the algorithm which assigns $x-z$ in phase $t$ is threshold-respecting on $W^{(t)}$.
Thus, we can apply \cref{lem:value_per_phase} to obtain the desired upper bound on $\1^\top(x-z)$.

\end{proof}

\subsection{Putting everything together}
In this section, we complete the proof of \cref{thm:fractional}. 
Since we assumed that the algorithm is $\varepsilon$-threshold respecting in the first $T-1$ phases, we can apply \Cref{cor_total_value_first_phases} to upper bound the value obtained in the first $T-1$ phases as
\[\sum_{t=1}^{T-1}\left(\frac{e-1}{e+1}m + \varepsilon t m + O\left((tm)^{2/3}\right) \right)  \leq \frac{e-1}{e+1} Tm + \varepsilon T^2m + O \left(T  ^{5/3}m^{2/3}\right).\]
By Theorem \ref{thm_bound_last_phase}, the value gained by the algorithm on each component $C_i$ during the last phase $T$ is at most $O(\sqrt{T} + \varepsilon T^2)$. Hence, the algorithm gains at most $O(\sqrt{T}m + \varepsilon T^2m)$ in the last phase.

We now argue that our instance $\mathcal{H}=(V,W,H)$ has a perfect matching.

\begin{lemma}
\label{lemma_OPT}
Our adversarial instance $\mathcal{H} = (V,W,H)$ satisfies $\opt(\mathcal{H}) = T m.$
\end{lemma}

\begin{proof} 
We prove that for every $t\in [T]$, there exists a hypergraph matching of size $tm$ at the end of phase $t$.
Let $C_i$ be a component with bipartition $U_i = [T]$ and $V_i = [T]$.
It suffices to show that there exists a graph matching $\widetilde \Match^{(t)}_i$ with vertex set $[t]$ on each side.
This is because $\widetilde \Match^{(t)}:= \cup_{i=1}^m \widetilde \Match^{(t)}_i$ can be extended to a hypergraph matching in $\mathcal{H}$ by our construction (see \Cref{lemma_vertex_arrival}).
Let $E^{(t)}_i\subseteq \Match^{(t)}_i$ be the edges whose endpoints are not active at the end of phase $t$.
Then, a simple inductive argument on $t\geq 1$ shows that $\cup_{s=1}^{t-1} E^{(s)}_i \cup \Match^{(t)}_i$ is a graph matching with vertex set $[t]$ on each side (see Figure \ref{fig_edge_arr_w_threshold} for an example).
\end{proof}

By Lemma \ref{lemma_OPT}, the competitive ratio of the algorithm is at most 
\[\frac{e-1}{e+1} + O(\varepsilon T + T^{2/3}m^{-1/3}+T^{-1/2}).\] 
Hence, letting $m \to \infty$, picking $T = o(\sqrt{m})$ such that $T\to \infty$ and setting $\varepsilon = o(1/T)$, we conclude that the competitive ratio is upper bounded by $(e-1)/(e+1)$, thus finishing the proof of \cref{thm:fractional}.

\section{Integral algorithm for bounded degree hypergraphs}
\label{app:bounded_degree_algo}

In this section, we show that {\sc Random} (\cref{bounded_degree_algo}) performs better than the greedy algorithm when the online nodes have bounded degree.

\begin{algorithm}
    \caption{{\sc Random} algorithm for bounded degree hypergraphs}
    \label{bounded_degree_algo}
    \begin{algorithmic}
    \State $\mathbf{Input: }$ $k$-uniform hypergraph $\mathcal{H} = (V,W,H)$ with online arrivals of each $w \in W$ with $|\delta(w)| \leq d$.
    \State $\mathbf{Output: }$ Matching $\Match \subseteq H$
    \vspace{0.2cm}
    \State set $\mathcal{M} \gets \emptyset$
    \State \textbf{when $w \in W$ arrives with $\delta(w) \subseteq H$:}
     \State \qquad pick uniformly at random $h\in\delta(w)$ among the hyperedges that are disjoint from $\Match$
     \State \qquad set $y_v = \min\left(\frac{1}{k-1},\frac{d}{(d-1)k+1}\right)$ for all $v\in h\setminus \{w\}$
     \State \qquad set $y_w = \max\left(0,\frac{d-k+1}{(d-1)k+1}\right)$
\State \Return $\Match$
    \end{algorithmic}
    \end{algorithm}

We prove \cref{thm:integral_bounded_degree}, restated below.

\begin{theorem}
    Algorithm \ref{bounded_degree_algo} is $\rho$-competitive for $k$-uniform hypergraphs whose online nodes have degree at most $d$, where
    \begin{align*}
        \rho = \min\left(\frac1{k-1}, \frac{d}{(d-1)k + 1}\right).
    \end{align*}
\end{theorem}
\begin{proof}
Let the algorithm be denoted by $\mathcal{A}$. We prove the result via a primal-dual analysis, where the random primal solution is given by $x_h := \mathbbm{1}_{\{h \in \Match\}}$ for every $h \in H$ and the random dual solution is the vector $y \in [0,1]^{V\cup W}$ constructed during the execution of the algorithm. Observe that the objective values of both solutions are equal at all times during the execution of the algorithm:
\begin{align}
\label{eq:primal_dual_eq}
\val(\mathcal{A}) = |\Match| = \sum_{h \in H}x_h = \sum_{v \in V \cup W} y_v.
\end{align}
This holds since every time a hyperedge $h \in H$ is matched by the algorithm, increasing the primal value $\val(\mathcal{A})$ by one, the dual objective increases by $\sum_{v \in h} y_v$. Two easy computations that we omit show that the latter is also equal to one in both cases where $d \leq k-1$ and $d \geq k-1$.

    We will now show that, in expectation, the dual constraints are satisfied up to a factor of $\rho$, i.e. 
    \begin{align}
        \Exp\left[\sum_{v\in h}y_v\right] \geq \rho \qquad \forall h \in H. \label{eq:d_bound_dual}
    \end{align}
This will imply the theorem, since the random vector $\Exp[y]/\rho$ will then be a feasible dual solution, leading to $\Exp[\val(\mathcal{A})] = \Exp\left[\sum_{v\in V \cup W}y_v\right] \geq \rho \: \optLP$ by \eqref{eq:primal_dual_eq} and \eqref{eq:d_bound_dual}.

    To show this inequality, let $h \in H$ be an arbitrary hyperedge incident to some online node $w \in W$. We now consider the following probabilistic event upon the arrival of $w$:
\[\mathcal{E} := \Big\{ \exists v \in h \setminus \{w\} \text{ which is already matched at the arrival of $w$}\Big\}.\]
 We will show \eqref{eq:d_bound_dual} by conditioning on $\mathcal{E}$ and on its complementary event $\bar{\mathcal{E}}$, which states that all nodes in $h \setminus \{w\}$ are unmatched when $w$ arrives, and that the hyperedge $h$ is thus available and considered in the random choice of the algorithm in this step. In the first case, if $\mathcal{E}$ happens, then some offline node $u \in h\setminus \{w\}$ has already had its dual value set to $y_u = \min\left(\frac{1}{k-1},\frac{d}{(d-1)k+1}\right) = \rho$ in a previous step of the algorithm, leading to
 \[ \Exp\left[\sum_{v\in h} y_v \: \Big| \: \mathcal{E} \right] = \sum_{v\in h} \Exp\left[y_v \mid \mathcal{E} \right] \geq \Exp\big[y_u \mid \mathcal{E}\big] = \rho. \]
Otherwise, if $\bar{\mathcal{E}}$ happens, we know that with probability at least $1/d$, the algorithm adds $h$ to the matching. Summing the dual values of the offline nodes contained in $h$ gives
    \begin{align*}
        \sum_{v\in h\setminus \{w\}}\Exp\left[y_v\mid \bar{\mathcal{E}}\:\right] &\geq \frac{1}{d}\cdot (k-1) \cdot \min\left(\frac{1}{k-1},\frac{d}{(d-1)k+1}\right).
    \end{align*}
    Furthermore, since the algorithm will always match $w$ to a hyperedge in this case, we have:
    \begin{align*}
        \Exp\left[y_w\mid \bar{\mathcal{E}} \: \right] &=  \max\left(0,\frac{d-k+1}{(d-1)k+1}\right).
    \end{align*}
    Adding those terms together, we get:
    \begin{align*}
        \sum_{v\in h}\Exp\left[y_v\mid \bar{\mathcal{E}} \: \right] &\geq \frac{1}{d}\cdot (k-1) \cdot \min\left(\frac{1}{k-1},\frac{d}{(d-1)k+1}\right) + \max\left(0,\frac{d-k+1}{(d-1)k+1}\right)\\
        &= \min\left(\frac{1}{d},\frac{k-1}{(d-1)k+1}\right) + \max\left(0,\frac{d-(k-1)}{(d-1)k+1}\right)\\
        &\geq \min\left(\frac{1}{d},\frac{k-1}{(d-1)k+1}+ \frac{d-(k-1)}{(d-1)k+1} \right) \geq \rho.
    \end{align*}
    This shows that \eqref{eq:d_bound_dual} holds, and hence proves that the algorithm is $\rho$-competitive.
\end{proof}

\section{Concluding remarks}
Our main contribution was a tight result for the online fractional matching problem under vertex arrivals for $3$-uniform hypergraphs, with matching upper and lower bounds of $(e-1)/(e+1) \approx 0.46$. The biggest open question raised by this work is the integral setting, where the best known algorithm is still the $1/3$-competitive simple greedy algorithm. Possible directions to get an improvement would be to develop new rounding techniques of fractional algorithms on hypergraphs, or by trying to find an analogue of the randomized \Call{Ranking}{} algorithm for bipartite graphs. This hypergraph setting adds additional difficulties to the online matching problem, for which new techniques seem to be required. Another interesting direction would be to study this problem in different online arrival models, for instance edge-arrivals.

For $k$-uniform hypergraphs with $k \geq 3$ under vertex arrivals, the integral setting admits upper and lower bounds of $1/k$ and $2/k$, whereas the fractional setting has upper and lower bounds of $\Theta(1/\log k)$. It would still be interesting to derive the optimal closed form solution depending on $k$ for both of these settings, under either vertex or edge arrivals. Note that for large $k$, both models are very similar, as there is a simple reduction from $k$-uniform hypergraphs under edge-arrivals to $(k+1)$-uniform hypergraphs under vertex-arrivals by appending degree one online nodes.

For both the integral and fractional settings, an additional direction would be to study the variant in which previously added hyperedges can be removed from the matching. This is a natural question, as both our upper bound constructions for the fractional and integral settings rely on the fact that an algorithm cannot revoke previous decisions.

\bibliographystyle{plain}
\bibliography{references}

@inproceedings{gamlath2019online,
  title={Online matching with general arrivals},
  author={Gamlath, Buddhima and Kapralov, Michael and Maggiori, Andreas and Svensson, Ola and Wajc, David},
  booktitle={60th {IEEE} Annual Symposium on Foundations of Computer Science {(FOCS)}},
  pages={26--37},
  year={2019},
  organization={IEEE}
}

@inproceedings{karp1990optimal,
  title={An optimal algorithm for on-line bipartite matching},
  author={Karp, Richard M and Vazirani, Umesh V and Vazirani, Vijay V},
  booktitle={Proceedings of the twenty-second annual ACM symposium on Theory of computing},
  pages={352--358},
  year={1990}
}

@article{birnbaum2008line,
  title={On-line bipartite matching made simple},
  author={Birnbaum, Benjamin and Mathieu, Claire},
  journal={Acm Sigact News},
  volume={39},
  number={1},
  pages={80--87},
  year={2008},
  publisher={ACM New York, NY, USA}
}

@inproceedings{devanur2013randomized,
  title={Randomized primal-dual analysis of ranking for online bipartite matching},
  author={Devanur, Nikhil R and Jain, Kamal and Kleinberg, Robert D},
  booktitle={Proceedings of the twenty-fourth annual ACM-SIAM symposium on Discrete algorithms},
  pages={101--107},
  year={2013},
  organization={SIAM}
}

@inproceedings{goel2008online,
  title={Online budgeted matching in random input models with applications to Adwords.},
  author={Goel, Gagan and Mehta, Aranyak},
  booktitle={SODA},
  volume={8},
  pages={982--991},
  year={2008}
}

@inproceedings{eden2021economics,
  title={An Economics-Based Analysis of RANKING for Online Bipartite Matching},
  author={Eden, Alon and Feldman, Michal and Fiat, Amos and Segal, Kineret},
  booktitle={Symposium on Simplicity in Algorithms (SOSA)},
  pages={107--110},
  year={2021},
  organization={SIAM}
}

@inproceedings{karande2011online,
  title={Online bipartite matching with unknown distributions},
  author={Karande, Chinmay and Mehta, Aranyak and Tripathi, Pushkar},
  booktitle={Proceedings of the forty-third annual ACM symposium on Theory of computing},
  pages={587--596},
  year={2011}
}

@inproceedings{mahdian2011online,
  title={Online bipartite matching with random arrivals: an approach based on strongly factor-revealing lps},
  author={Mahdian, Mohammad and Yan, Qiqi},
  booktitle={Proceedings of the forty-third annual ACM symposium on Theory of computing},
  pages={597--606},
  year={2011}
}

@inproceedings{kesselheim2013optimal,
  title={An optimal online algorithm for weighted bipartite matching and extensions to combinatorial auctions},
  author={Kesselheim, Thomas and Radke, Klaus and T{\"o}nnis, Andreas and V{\"o}cking, Berthold},
  booktitle={European symposium on algorithms},
  pages={589--600},
  year={2013},
  organization={Springer}
}

@inproceedings{aggarwal2011online,
  title={Online vertex-weighted bipartite matching and single-bid budgeted allocations},
  author={Aggarwal, Gagan and Goel, Gagan and Karande, Chinmay and Mehta, Aranyak},
  booktitle={Proceedings of the twenty-second annual ACM-SIAM symposium on Discrete Algorithms},
  pages={1253--1264},
  year={2011},
  organization={SIAM}
}

@article{huang2019online,
  title={Online vertex-weighted bipartite matching: Beating 1-1/e with random arrivals},
  author={Huang, Zhiyi and Tang, Zhihao Gavin and Wu, Xiaowei and Zhang, Yuhao},
  journal={ACM Transactions on Algorithms (TALG)},
  volume={15},
  number={3},
  pages={1--15},
  year={2019},
  publisher={ACM New York, NY, USA}
}

@article{fahrbach2022edge,
  title={Edge-weighted online bipartite matching},
  author={Fahrbach, Matthew and Huang, Zhiyi and Tao, Runzhou and Zadimoghaddam, Morteza},
  journal={Journal of the ACM},
  volume={69},
  number={6},
  pages={1--35},
  year={2022},
  publisher={ACM New York, NY}
}

@inproceedings{buchbinder2007online,
  title={Online primal-dual algorithms for maximizing ad-auctions revenue},
  author={Buchbinder, Niv and Jain, Kamal and Naor, Joseph},
  booktitle={European Symposium on Algorithms},
  pages={253--264},
  year={2007},
  organization={Springer}
}

@inproceedings{devanur2012online,
  title={Online matching with concave returns},
  author={Devanur, Nikhil R and Jain, Kamal},
  booktitle={Proceedings of the forty-fourth annual ACM symposium on Theory of computing},
  pages={137--144},
  year={2012}
}

@inproceedings{huang2020adwords,
  title={Adwords in a panorama},
  author={Huang, Zhiyi and Zhang, Qiankun and Zhang, Yuhao},
  booktitle={2020 IEEE 61st Annual Symposium on Foundations of Computer Science (FOCS)},
  pages={1416--1426},
  year={2020},
  organization={IEEE}
}

@article{mehta2007adwords,
  title={Adwords and generalized online matching},
  author={Mehta, Aranyak and Saberi, Amin and Vazirani, Umesh and Vazirani, Vijay},
  journal={Journal of the ACM (JACM)},
  volume={54},
  number={5},
  pages={22--es},
  year={2007},
  publisher={ACM New York, NY, USA}
}

@article{kalyanasundaram2000optimal,
  title={An optimal deterministic algorithm for online b-matching},
  author={Kalyanasundaram, Bala and Pruhs, Kirk R},
  journal={Theoretical Computer Science},
  volume={233},
  number={1-2},
  pages={319--325},
  year={2000},
  publisher={Elsevier}
}

@inproceedings{marinkovic_online_2023,
  author       = {Javier Marinkovic and
                  Jos{\'{e}} A. Soto and
                  Victor Verdugo},
  title        = {Online Combinatorial Assignment in Independence Systems},
  booktitle    = {Integer Programming and Combinatorial Optimization - 25th International
                  Conference {(IPCO)}},
  pages        = {294--308},
  publisher    = {Springer},
  year         = {2024},
  address = {Berlin, Heidelberg}
}

@article{huang_fully_online_1,
  author       = {Zhiyi Huang and
                  Ning Kang and
                  Zhihao Gavin Tang and
                  Xiaowei Wu and
                  Yuhao Zhang and
                  Xue Zhu},
  title        = {Fully Online Matching},
  journal      = {J. {ACM}},
  volume       = {67},
  number       = {3},
  pages        = {17:1--17:25},
  year         = {2020}
}

@inproceedings{huang_fully_online_2,
  author       = {Zhiyi Huang and
                  Zhihao Gavin Tang and
                  Xiaowei Wu and
                  Yuhao Zhang},
  title        = {Fully Online Matching {II:} Beating Ranking and Water-filling},
  booktitle    = {61st {IEEE} Annual Symposium on Foundations of Computer Science {(FOCS)}},
  pages        = {1380--1391},
  year         = {2020}
}

@article{mehta_survey,
  author       = {Aranyak Mehta},
  title        = {Online Matching and Ad Allocation},
  journal      = {Found. Trends Theor. Comput. Sci.},
  volume       = {8},
  number       = {4},
  pages        = {265--368},
  year         = {2013}
}

@article{ashlagi_online_windowed,
  author       = {Itai Ashlagi and
                  Maximilien Burq and
                  Chinmoy Dutta and
                  Patrick Jaillet and
                  Amin Saberi and
                  Chris Sholley},
  title        = {Edge-Weighted Online Windowed Matching},
  journal      = {Math. Oper. Res.},
  volume       = {48},
  number       = {2},
  pages        = {999--1016},
  year         = {2023}
}

@article{buchbinder2009online,
  title={Online primal-dual algorithms for covering and packing},
  author={Buchbinder, Niv and Naor, Joseph},
  journal={Mathematics of Operations Research},
  volume={34},
  number={2},
  pages={270--286},
  year={2009},
  publisher={INFORMS}
}

@inproceedings{conf/icalp/KorulaP09,
  author       = {Nitish Korula and
                  Martin P{\'{a}}l},
  title        = {Algorithms for Secretary Problems on Graphs and Hypergraphs},
  booktitle    = {Automata, Languages and Programming, 36th Internatilonal Colloquium
                  {(ICALP)}, Proceedings, Part {II}},
  pages        = {508--520},
  publisher    = {Springer},
  year         = {2009},
  address = {Berlin, Heidelberg}
}

@article{journals/ior/PavoneSST22,
  author       = {Marco Pavone and
                  Amin Saberi and
                  Maximilian Schiffer and
                  Matt Wu Tsao},
  title        = {Technical Note - Online Hypergraph Matching with Delays},
  journal      = {Oper. Res.},
  volume       = {70},
  number       = {4},
  pages        = {2194--2212},
  year         = {2022}
}

@article{journals/ior/MaRST20,
  author       = {Yuhang Ma and
                  Paat Rusmevichientong and
                  Mika Sumida and
                  Huseyin Topaloglu},
  title        = {An Approximation Algorithm for Network Revenue Management Under Nonstationary
                  Arrivals},
  journal      = {Oper. Res.},
  volume       = {68},
  number       = {3},
  pages        = {834--855},
  year         = {2020}
}

@inproceedings{yao1977probabilistic,
  title={Probabilistic computations: Toward a unified measure of complexity},
  author={Yao, Andrew Chi-Chin},
  booktitle={18th Annual Symposium on Foundations of Computer Science (sfcs 1977)},
  pages={222--227},
  year={1977},
  organization={IEEE Computer Society}
}

@inproceedings{conf/icalp/WangW15,
  author       = {Yajun Wang and
                  Sam Chiu{-}wai Wong},
  title        = {Two-sided Online Bipartite Matching and Vertex Cover: Beating the
                  Greedy Algorithm},
  booktitle    = {Automata, Languages, and Programming - 42nd International Colloquium
                  {(ICALP)}, Proceedings, Part {I}},
  pages        = {1070--1081},
  publisher    = {Springer},
  year         = {2015},
  address = {Berlin, Heidelberg}
}

@article{journals/algorithmica/BuchbinderST19,
  author       = {Niv Buchbinder and
                  Danny Segev and
                  Yevgeny Tkach},
  title        = {Online Algorithms for Maximum Cardinality Matching with Edge Arrivals},
  journal      = {Algorithmica},
  volume       = {81},
  number       = {5},
  pages        = {1781--1799},
  year         = {2019}
}

@article{chan_linear_2012,
  title={On linear and semidefinite programming relaxations for hypergraph matching},
  author={Chan, Yuk Hei and Lau, Lap Chi},
  journal={Mathematical programming},
  volume={135},
  number={1-2},
  pages={123--148},
  year={2012},
  publisher={Springer}
}

@inproceedings{conf/focs/Cygan13,
  author       = {Marek Cygan},
  title        = {Improved Approximation for 3-Dimensional Matching via Bounded Pathwidth
                  Local Search},
  booktitle    = {54th Annual {IEEE} Symposium on Foundations of Computer Science {(FOCS)}},
  pages        = {509--518},
  publisher    = {{IEEE} Computer Society},
  year         = {2013},
  address = {Los Alamitos, CA, USA}
}

@article{journals/orl/TrobstU24,
  author       = {Thorben Tr{\"{o}}bst and
                  Rajan Udwani},
  title        = {Almost tight bounds for online hypergraph matching},
  journal      = {Oper. Res. Lett.},
  volume       = {55},
  pages        = {107143},
  year         = {2024}
}

@inproceedings{hutchison_b-matching_2012,
  title={The b-Matching Problem in Hypergraphs: Hardness and Approximability},
  author={El Ouali, Mourad and J{\"a}ger, Gerold},
  booktitle={International Conference on Combinatorial Optimization and Applications},
  pages={200--211},
  year={2012},
  organization={Springer}
}

@incollection{doerr_probabilistic_2020-1,
  title={Probabilistic tools for the analysis of randomized optimization heuristics},
  author={Doerr, Benjamin},
  booktitle={Theory of evolutionary computation: Recent developments in discrete optimization},
  pages={1--87},
  year={2019},
  editor = {Doerr, Benjamin and Neumann, Frank},
  publisher = {Springer},
}

@article{raghavan_randomized_1987-2,
  title={Randomized rounding: a technique for provably good algorithms and algorithmic proofs},
  author={Raghavan, Prabhakar and Tompson, Clark D},
  journal={Combinatorica},
  volume={7},
  number={4},
  pages={365--374},
  year={1987},
  publisher={Springer}
}

@article{srivastav_weighted_1995,
  title={Weighted fractional and integral k-matching in hypergraphs},
  author={Srivastav, Anand and Stangier, Peter},
  journal={Discrete applied mathematics},
  volume={57},
  number={2-3},
  pages={255--269},
  year={1995},
  publisher={Elsevier}
}

@inproceedings{lee_asymptotically_2025,
  title={Asymptotically Optimal Hardness for k-Set Packing and k-Matroid Intersection},
  author={Lee, Euiwoong and Svensson, Ola and Thiery, Theophile},
  booktitle={Proceedings of the 57th Annual ACM Symposium on Theory of Computing},
  pages={54--61},
  year={2025}
}

\appendix
\section{Proof of Lemma \ref{lemma_vertex_arrival}}
\label{app:proofs_upper_bound}
\begin{proof}
Let us fix a fractional algorithm $\mathcal{A}$ and let us fix a matching $\Match = (V, E)$ of size $n$, meaning that $|E|=n$.
The adversarial online 3-uniform hypergraph instance $\mathcal{H}$ consists of $n$ online nodes $W = \{w_1, \dots ,w_n\}$ arriving and connecting to a subset of edges of the matching $\Match$.
For every $w_i$, we denote by $E(w_i) \subseteq E$ the edges of the matching the online node $w_i$ is connected to, meaning that the 3-hyperedges incident to $w_i$ are $\delta(w_i) = \{w_i \cup e : e \in E(w_i)\}$.
Let us denote by $x \in \mathbb{R}^E$ the fractional solution generated online by algorithm $\mathcal{A}$, and note that this is in fact the induced fractional solution on $E$.
\begin{enumerate}
\item The first online node $w_1$ connects to every edge of the matching, i.e. $E(w_1) = E$. The algorithm $\mathcal{A}$ now assigns fractional value $x(e)$ to every edge $e \in E$, and we denote by $e_1 \in E$ the edge with the lowest fractional value $x(e_1)$. Observe that $x(e_1) \leq 1/n$.
\item The second online node $w_2$ connects to $E(w_2) = E \setminus \{e_1\}$. The algorithm $\mathcal{A}$ can thus increase the fractional values $x(e)$ for every $e \in E(w_2)$. We then denote by $e_2$ the edge in $E(w_2)$ with the lowest fractional value after this iteration, and it is easy to check that $x(e_1) + x(e_2) \leq 2/n + 1/(n-1)$.
\item More generally, for every $k \in \{1, \dots, n\}$, the online node $w_k$ connects to $n-k+1$ edges $E(w_k) = E \setminus \{e_1, \dots, e_{k-1}\}$, and $e_k$ is defined as the edge having the lowest fractional value at the end of the iteration of $w_k$. We thus get a bound of
\begin{equation}
\label{eq_bound_vertex_arrival}
\sum_{k=1}^{\ell} x(e_k) \leq \sum_{k=1}^{\ell} \sum_{i = 1}^{k} \frac{1}{n-i+1} \qquad \forall \ell \in \{1, \dots, n\}.
\end{equation}
\end{enumerate}
The inner sum in \eqref{eq_bound_vertex_arrival} reaches $\Delta$ approximately when $k \approx (1 - e^{-\Delta})n$. For higher values of $k$, it is thus better to use the bound $x(e_k) \leq \Delta$, which holds by assumption. By defining $p := \lfloor{e^{-\Delta}n}\rfloor$ and $q := n - p$, we can now compute a precise upper bound on the total value generated by the algorithm using \eqref{eq_bound_vertex_arrival}:
\begin{align}
\label{eq_value_up_bound}
\val(A, \mathcal{H}) &= \sum_{k = 1}^{n}x(e_k) \leq \sum_{k=1}^{q}\sum_{i = 1}^{k} \frac{1}{n-i+1} + \sum_{k = q+1}^{n}\Delta = \sum_{i=1}^{q} \sum_{k = i}^{q} \frac{1}{n-i+1} + p \Delta \nonumber \\
&= \sum_{i = 1}^{q}\frac{q - i +1}{n-i+1} + p \Delta = q - (n-q)\sum_{i=1}^q \frac{1}{n-i+1} + p \Delta \nonumber \\
&= p\Delta + (n-p) - p \sum_{i = n-q+1}^{n}\frac{1}{i} =  p \Delta + n - p - p (H_n - H_p)
\end{align}
In order to get the desired result for every value of $n \geq 1$, we now need to tightly approximate the difference of the harmonic numbers $H_n - H_p$. In particular, the well known bounds $\ln(n) + 1/n \leq H_n \leq \ln(n+1)$ for every $n \in \mathbb{N}$ are not enough in this case. We use the equality
\begin{equation}
\label{eq_Hn_approx}
H_n = \ln(n) + \gamma + \epsilon(n) \qquad \text{for some } 0 < \epsilon(n) < \frac{1}{2n}
\end{equation}
where $\gamma = \lim_{n \to \infty} (H_n - \ln(n)) \approx 0.58$ is Euler's constant. Moreover, recall that 
\begin{equation}
\label{eq_p_approx}
e^{-\Delta}n - 1 \leq p \leq e^{-\Delta}n.
\end{equation}
Using \eqref{eq_Hn_approx} and \eqref{eq_p_approx} together gives:
\begin{align}
\label{eq_diff_harmonic_nb}
H_n - H_p &= \ln\left(\frac{n}{p}\right) + \epsilon(n) - \epsilon(p) \geq \ln\left(\frac{n}{e^{-\Delta}n}\right) - \frac{1}{2p} = \Delta - \frac{1}{2p} \nonumber \\
\end{align}
Finally, plugging \eqref{eq_p_approx} and \eqref{eq_diff_harmonic_nb} into \eqref{eq_value_up_bound} gets us the desired result for every value of $n \in \mathbb{N}:$
\begin{align*}
\val(A, \mathcal{H}) &\leq p \Delta + n - p - p \left(\Delta - \frac{1}{2p}\right) = n - p + \frac{1}{2} \leq (1 - e^{-\Delta})n + \frac{3}{2}.
\end{align*}
\end{proof}

\section{Justification of assumptions in \Cref{subsub:analysis_overview}}
\label{sec:assumptions}

In this section, we justify the two assumptions made on the algorithm in \cref{subsub:analysis_overview}.

\subsection{Assumption 1: Symmetry} 
We start by justifying the symmetry assumption.
For a vertex-arrival hypergraph $\mathcal{H}$, we denote $V(\mathcal{H})$ as the set of offline nodes, $W(\mathcal{H})$ as the set of online nodes, and $H(\mathcal{H})$ as the set of hyperedges.

\begin{definition}
For a vertex-arrival hypergraph $\mathcal{H} = (V,W,H)$, an \emph{automorphism} is a permutation $\sigma$ of the offline nodes $V$ such that for every $S\subseteq V$ and $w\in W$, 
\[S\cup\{w\}\in H \iff \{\sigma(v):v\in S\}\cup\{w\}\in H.\]
\end{definition}
For a subset $S\subseteq V$, we write $\sigma(S) \coloneqq \{\sigma(v):v\in S\}$.
For a hyperedge $h = S\cup\{w\}$ where $S\subseteq V$ and $w\in W$, we also write $\sigma(h) \coloneqq \sigma(S)\cup \{w\}$.

\begin{definition}
    For a vertex-arrival hypergraph $\mathcal{H} = (V,W,H)$, let $\Sigma$ be a subset of its automorphisms. 
We say that a fractional matching $x$ of $\mathcal{H}$ is \emph{$\Sigma$-symmetric} if $x_h=x_{\sigma(h)}$ for all $h\in H$ and $\sigma\in \Sigma$.
\end{definition}

Since the construction of our vertex-arrival instance depends on the actions of the algorithm, we will overload the notation $\mathcal{H}$ as follows.
A \emph{(adaptive) vertex-arrival instance} is a function $\mathcal{H}$ which takes as input an algorithm $\Alg$ and outputs a vertex-arrival hypergraph $\mathcal{H}(\Alg)$.
For concreteness, the reader can think of $\mathcal{H}$ as the instance constructed in \Cref{sec:upper_bound}.
For $i\geq 1$, let $\mathcal{H}_i$ be the subinstance of $\mathcal{H}$ that ends with the arrival of the $i$th online node $w_i$.
We can assume that $\mathcal{H}_1(\Alg) = \mathcal{H}_1(\Alg')$ for any pair of algorithms $\Alg$ and $\Alg'$, as algorithms do not take any actions before the first online node arrives.
So, the subinstance $\mathcal{H}_1$ can be thought of as a hypergraph.

\begin{definition}
For a vertex-arrival instance $\mathcal{H}$ and an algorithm $\Alg$, let $\Sigma$ be a subset of automorphisms of $\mathcal{H}(\Alg)$.
We say that $\mathcal{A}$ is \emph{$\Sigma$-symmetric on $\mathcal{H}$} if it returns a $\Sigma$-symmetric fractional matching of $\mathcal{H}(\mathcal{A})$ given $\mathcal{H}$.
\end{definition}

\newcommand{\sym}{\operatorname{sym}}
  
Fix a vertex-arrival instance $\mathcal{H}$, and let $\Sigma$ be a subgroup of automorphisms of $\mathcal{H}_1$.
For an algorithm $\Alg$, we define its \emph{$\Sigma$-symmetrization} $\sym_\Sigma(\Alg)$ as follows.
We first run $\Alg$ on $\mathcal{H}_1$ to obtain a fractional matching $x'$.
Then, $\sym_\Sigma(\Alg)$ sets
\begin{align*}
        x_h:=\frac{1}{|\Sigma|}\sum_{\sigma\in \Sigma}x'_{\sigma(h)} \qquad \forall\; h\in H(\mathcal{H}_1).
\end{align*}
This in turn determines the next set of hyperedges in the hypergraph $\mathcal{H}_2(\sym_\Sigma(\Alg))$.
Note that $\mathcal{H}_2(\sym_\Sigma(\Alg))$ may not be equal to $\mathcal{H}_2(\Alg)$.
As long as $\Sigma$ remains a subgroup of automorphisms of $\mathcal{H}_2$, we repeat the process above.
So, for any $i\geq 2$, if $x'$ is the fractional matching obtained by running $\mathcal{A}$ on $\mathcal{H}_i(\sym_\Sigma(\Alg))$, then $\sym_\Sigma(\Alg)$ sets
\begin{align*}
        x_h:=\frac{1}{|\Sigma|}\sum_{\sigma\in \Sigma}x'_{\sigma(h)} \qquad \forall\; h\in H(\mathcal{H}_i(\sym_\Sigma(\Alg))).
\end{align*}
Note that $\sym_\Sigma(\Alg)$ is well-defined if and only if $\Sigma$ is a subgroup of automorphisms of the full hypergraph $\mathcal{H}(\sym_\Sigma(\Alg))$.

\begin{lemma}\label{lem:group}
If $\sym_\Sigma(\Alg)$ is well-defined, then it is $\Sigma$-symmetric on $\mathcal{H}$.
\end{lemma}

\begin{proof} 
Let $x$ be the fractional matching obtained by running $\sym_\Sigma(\Alg)$ on $\mathcal{H}(\sym_\Sigma(\Alg))$, and let $x'$ be the fractional matching obtained by running $\Alg$ on $\mathcal{H}(\sym_\Sigma(\Alg))$.
For any hyperedge $h$ in $\mathcal{H}(\sym_\Sigma(\mathcal{A}))$ and any permutation $\tau\in \Sigma$, we have
\[x_h = \frac{1}{|\Sigma|}\sum_{\sigma\in \Sigma}x'_{\sigma(h)} =  \frac{1}{|\Sigma|}\sum_{\sigma\in \Sigma}x'_{\sigma(\tau(h))} =  x_{\tau(h)},\]
where the second equality is due to $\Sigma$ being a group.
\end{proof}

The following lemma shows that in order to construct a worst-case instance for all algorithms, it suffices to construct a worst-case instance $\mathcal{H}$ for algorithms that are $\Sigma$-symmetric on $\mathcal{H}$.
That is, $\mathcal{H}$ can be extended to an instance $\mathcal{H}'$ for all algorithms on which every algorithm has the same performance as its $\Sigma$-symmetrization.

\begin{lemma}\label{lem:symmetric_algorithm}
    For a vertex-arrival instance $\mathcal{H}$, let $\Sigma$ be a subgroup of automorphisms of $\mathcal{H}_1$.
    Suppose that $\Sigma$ remains a subgroup of automorphisms of $\mathcal{H}_{i}(\Alg)$ for any $i\geq 2$ and any algorithm $\Alg$ that is $\Sigma$-symmetric on $\mathcal{H}_{i-1}$. 
    Then, there exists a vertex-arrival instance $\mathcal{H}'$ such that for every algorithm $\Alg$:
    \begin{itemize}
        \item  $\mathcal{H}'(\Alg) = \mathcal{H}(\sym_\Sigma(\Alg))$,
        \item $\val(\Alg,\mathcal{H}')= \val(\sym_\Sigma(\Alg),\mathcal{H})$.
    \end{itemize}
\end{lemma}

\begin{proof}
    Let $\Alg$ be any algorithm.
    Observe that $\sym_\Sigma(\Alg)$ is well-defined due to our assumption.
    Hence, $\sym_\Sigma(\Alg)$ is $\Sigma$-symmetric on $\mathcal{H}$ by \Cref{lem:group},. 
    We define $\mathcal{H}'(\Alg):=\mathcal{H}(\sym_\Sigma(\Alg))$.
    Let $x$ be the fractional matching obtained by running $\sym_\Sigma(\Alg)$ on $\mathcal{H}(\sym_\Sigma(\Alg))$, and let $x'$ be the fractional matching obtained by running $\Alg$ on $\mathcal{H}'(\Alg)$. 
    By definition, we have $x_h:=\frac{1}{|\Sigma|}\sum_{\sigma\in \Sigma}x'_{\sigma(h)}$ for every hyperedge $h\in H(\mathcal{H}(\sym_\Sigma(\Alg)))$. 
    Hence, 
    \[\sum_{h} x_h = \sum_{h}\frac{1}{|\Sigma|}\sum_{\sigma\in \Sigma}x'_{\sigma(h)} = \frac{1}{|\Sigma|} \sum_{\sigma\in \Sigma} \sum_{h}x'_{\sigma(h)} = \sum_{h}x'_{h}. \] \todo{CK: Added the 3rd term in the chain of equations.}
\end{proof}

In \cref{subsub:analysis_overview}, we assumed that the algorithm treats the $k$th vertex in $U_i$, say $u_{i,k}$, and the $k$th vertex in $V_i$, say $v_{i,k}$, symmetrically.
If our constructed hypergraph $\Hyp$ was symmetric with respect to these vertices, i.e. if the permutation $\sigma$ swapping $u_{i,k}$ and $v_{i,k}$ for all $i$ and $k$ was an automorphism of $\Hyp$, then \cref{lem:symmetric_algorithm} would show that this assumption can be made without loss of generality.
In particular, using the subgroup $\Sigma = (\{\sigma,e\},\circ)$ where $e$ is the identity permutation, it shows that $\mathcal{H}$ can be extended to all algorithms $\Alg$ so that $\Alg$ and $\sym_\Sigma(\Alg)$ have the same performance. \todo{CK: Added more expanation.}

However, one part of the instance that breaks this symmetry is the construction given in the proof of \cref{lemma_vertex_arrival} and illustrated in Figure \ref{fig_vertex_arrival}. As a reminder, this construction is repeatedly applied to submatchings of $\mathcal{M}^{(t)}$ in \cref{sec:firstphases}. Let us fix one such submatching and denote it by $\mathcal{M}:=\mathcal{M}^{(t)}(i,j)$.
As described in \cref{sec:lastphase} and illustrated in Figure \ref{fig_edge_arr_w_threshold}, if some $u_{i,k} \in \mathcal{M}$, then $v_{i,k} \in \mathcal{M}$ and the submatching is symmetric with respect to this pair, i.e. $\sigma(e) \in \Match$ for every edge $e$ in $\Match$. However, due to the Lemma \ref{lemma_vertex_arrival} construction, $e\cup \{w\}$ might be a hyperedge in $\Hyp$ for some online vertex $w$, while $\sigma(e\cup \{w\})=\sigma(e)\cup \{w\}$ might not be a hyperedge in $\Hyp$.

To fix this, the construction can be slightly tweaked in the following way. An important observation is that the horizontal edges in $\Match$ (between $u_{i,k}$ and $v_{i,k}$) are not isomorphic to any other edge in the hypergraph, whereas each of the diagonal edges (non-horizontal edges) are isomorphic to exactly one other edge in $\Match$.
For this reason, we can first apply the Lemma \ref{lemma_vertex_arrival} construction on just the horizontal edges of $\Match$.

We can then apply a slightly modified construction to the diagonal edges, where the pairs of isomorphic edges are treated in the same way. In the original construction, a newly arriving online vertex $w$ would be connected to all edges in $\Match$ that were incident to the previous online vertex, except for the one with the smallest fractional value. In the modified construction, we instead consider the online vertices in groups of two. For every two consecutive online vertices, we connect them to all edges in $\Match$ that were incident to the previous online vertex, except for the diagonal pair with the smallest total fractional value. This ensures that the symmetry between the diagonal edges is respected.

This modified construction would slightly worsen the upper bound in \cref{lemma_vertex_arrival}.
Let $\Match_{\text{hor}}$ be the set of horizontal edges in $\Match$ and $\Match_{\text{diag}}$ be the set of diagonal edges in $\Match$. By applying \cref{lemma_vertex_arrival} to the horizontal edges, we get that the value of the matching is at most $(1 - e^{-\Delta})|\Match_{\text{hor}}| + 3/2$. The value of the diagonal edges is at most twice the value of the original construction from \cref{lemma_vertex_arrival} applied to a transversal of the pairs of diagonal edges, which is at most $(1 - e^{-\Delta})\cdot \frac{1}{2}|\Match_{\text{diag}}| + 3/2$ by \cref{lemma_vertex_arrival}. So the total value of the matching is at most:
\begin{align*}
    (1 - e^{-\Delta})|\Match_{\text{hor}}| + 3/2 + 2\cdot \left((1 - e^{-\Delta})\cdot \frac{1}{2}|\Match_{\text{diag}}| + 3/2\right)=(1 - e^{-\Delta})|\Match| + 9/2.
\end{align*}

This results in a constant of $9/2$ instead of $3/2$ in \cref{lemma_vertex_arrival},  and a constant of $42$ instead of $15$ in \cref{lemma_vertex_arrival_blocks,lem:value_per_phase,cor_total_value_first_phases}. This does not affect the asymptotic upper bound for large $m$. Hence, it shows that \cref{thm:fractional} also holds for non-symmetric algorithms.

\subsection{Assumption 2: There is an optimal \texorpdfstring{$\varepsilon$}{ε}-threshold respecting algorithm}
Next, we justify that we restrict to $\varepsilon$-threshold respecting algorithms in our proof.
Let $\mathcal{H}$ be the instance constructed in \Cref{subsec:construction_overview}. 
Let $f$ be the function given by
\[f(x) \coloneqq \frac{e^x}{e+1},\]
and recall the definition of $\varepsilon$-threshold respecting with respect to $f$ (\cref{def_thresh_resp}). 

For any algorithm $\Alg$ and $\varepsilon>0$, we now show that there exists an algorithm $\Alg'$ which is $\varepsilon$-threshold respecting on all online nodes before the last phase.
Moreover, there exists an instance $\Hyp'$ such that the performance of $\Alg$ on $\Hyp'$ matches the performance of $\Alg'$ on $\Hyp$.

\begin{lemma}\label{lem:threshold_respecting}
Let $\mathcal{H}$ be the instance constructed in \Cref{subsec:construction_overview}.
For any algorithm $\Alg$ and $\varepsilon>0$, there exists an algorithm $\Alg'$ which is $\varepsilon$-threshold respecting on all online nodes before the last phase.
Furthermore, there exists an instance $\Hyp'$ such that
\[\frac{\val(\Alg,\mathcal{H}'(\Alg))}{\opt(\Hyp'(\Alg))} = \frac{\val(\Alg',\Hyp(\Alg'))}{\opt(\Hyp(\Alg'))}.\]
\end{lemma}

\begin{proof}
From $\Hyp$, we construct a new instance $\Hyp'$ as follows.
Let $N = \lceil 2/\varepsilon \rceil$. 
For every offline node $v$ in $\Hyp$, create $N$ offline copies in $\Hyp'$, denoted $v'_1,v'_2,\dots,v'_N$.
The new algorithm $\Alg'$ will be defined based on the behaviour of $\Alg$ on $\Hyp'$.
When the $i$th online node $w_i$ arrives in $\Hyp(\Alg')$ for $i\geq 1$, at most $N$ copies of $w_i$ arrives sequentially in $\Hyp'(\Alg)$, denoted $w'_{i,1},w'_{i,2},\dots$.
When the $j$th copy $w'_{i,j}$ arrives, for every edge $h = S\cup w_i$ in $\Hyp_i(\Alg')$, add the edge $h'_j \coloneqq \{v'_j:v\in S\}\cup w'_{i,j}$ to $\Hyp'(\Alg)$.
Now, let $x'_{i,j}$ denote the solution given by $\Alg$ in $\Hyp'$ after the arrival of $w'_{i,j}$.
Consider the following averaged solution
\[x_{i,j}(h) \coloneqq \frac{1}{N} \sum_k x'_{i,j}(h'_k) \qquad \qquad \forall h\in H(\Hyp_i(\Alg')).\]
If $j=N$, or $w_i$ appeared before the last phase and there exists a hyperedge $h\in \delta(w_i)$ such that
\[\sum_{v\in h\setminus \{w_i\}} f(x_{i,j}(\delta(v))) \geq 1,\]
then $w'_{i,j+1},\dots,w'_{i,N}$ will not arrive in $\Hyp'$.
In this case, $\Alg'$ sets $x(h) \gets x_{i,j}(h)$ for all $h\in \delta(w_i)$ in $\Hyp$.
Otherwise, we proceed to let the $(j+1)$th copy $w'_{i,j+1}$ arrive in $\Hyp'$.
This completes the description of $\Alg'$.

Clearly, $x$ is a fractional matching in $\Hyp(\Alg')$. Moreover,
\[\val(\Alg',\mathcal{H}(\Alg')) = \sum_h x(h) = \frac{\val(\Alg,\mathcal{H}'(\Alg))}{N}.\]
Next, we claim that $N \cdot \opt(\Hyp(\Alg')) = \opt(\Hyp'(\Alg))$.
Based on the construction of $\Hyp$, the offline optimal matching in $\Hyp(\Alg')$ covers all the online nodes in the last phase, and the online nodes on which $\Alg'$ is strictly threshold respecting. Let $W_\opt$ denote the union of these two sets.
For each $w_i\in W_\opt$, observe that $w_{i,j}$ is present in $\Hyp'(\Alg)$ for all $j\in [N]$ by our construction of $\Hyp'$.
Hence, the offline optimal matching in $\Hyp'(\Alg)$ covers the following online nodes
\[\{w_{i,j}: w_i\in W_\opt, j\in [N]\}.\]
So, $N \cdot \opt(\Hyp(\Alg')) = \opt(\Hyp'(\Alg))$ as desired.

It is left to show that $\Alg'$ is $\varepsilon$-threshold respecting on all online nodes before the last phase.
Pick such an online node $w_i$ and let $w_{i,j}$ be its last copy in $\mathcal{H'}(\Alg)$.
Note that $x_{i,j}$ is the output of $\Alg'$ in $\Hyp_i(\Alg')$.
For any $h\in \delta(w_i)$, we have
\begin{align*}
\sum_{v\in h\setminus\{w_i\}}f(x_{i,j}(\delta(v))) &\leq \sum_{v\in h\setminus\{w_i\}} f\left(x_{i,j-1}(\delta(v))+\frac1N \right) \tag{$x'_{i,j}(\delta(w'_{i,j}))\leq 1$} \\
&\leq \sum_{v\in h\setminus\{w_i\}} \left( f(x_{i,j-1}(\delta(v))+\frac{1}{N}  \right) \tag{$f$ is 1-Lipschitz}\\
&< 1 + \frac{2}{N} \tag{due to $|h| = 3$ and the construction of $\Hyp'$}\\
&\leq 1+\varepsilon
\end{align*}
\end{proof}
\section{Properties of $\psi(t,y)$}
\label{app:properties_psi}
\begin{proof}[Proof of \cref{lemma_psi}]
    The first two statement follow directly from the definition and the symmetry of $B(t,\frac12)$ around $\frac{t}{2}$. For the second statement, let $X\sim B(t,\frac12)$ and $Y\sim B(1,\frac12)$ be independent. Then, we have:
    \begin{align*}
        \psi(t+1,y)&=\Pr\left[X+Y< \frac{t+1}{2}+y\right] + \frac12\Pr\left[X+Y= \frac{t+1}{2}+y\right]\\
        &=\Pr[Y=1]\Pr\left[X< \frac{t}{2}+y-\frac12\right] + \Pr[Y=0]\Pr\left[X< \frac{t}{2}+y+\frac12\right]\\
        &\quad + \frac12\Pr[Y=0]\Pr\left[X= \frac{t}{2}+y+\frac12\right] + \frac12\Pr[Y=1]\Pr\left[X= \frac{t}{2}+y-\frac12\right]\\
        &=\frac12\psi\left(t,y-\frac12\right) + \frac12\psi\left(t,y+\frac12\right).
    \end{align*}
    Now, let us prove the last statement. Let $X\sim B(t,\frac12)$ and observe that $1-\psi(t,y)\leq \Pr[X\geq \frac{t}{2} +y]$, leading to:
    \begin{align*}
        \sum_{y=0}^\infty (1-\psi(t,\frac12 y))&\leq \sum_{y=0}^\infty \Pr\left[X-\frac{t}{2}\geq \frac{y}{2}\right]
        \leq \sum_{y=0}^\infty \Pr\left[\left|X-\frac{t}{2}\right|\geq \frac{y}{2}\right]\\&= \sum_{y=0}^\infty (y+1)\Pr\left[\left|X-\frac{t}{2}\right|= \frac{y}{2}\right] =
        1+2\sum_{y=0}^\infty \frac{y}{2}\Pr\left[\left|X-\frac{t}{2}\right|= \frac{y}{2}\right] \\
        &= 1+2 \: \Exp\left[\left|X-\frac{t}{2}\right|\right]
        \leq 1 + 2\sqrt{\Exp\left[\left(X-\frac{t}{2}\right)^2\right]}\qquad \text{(by Jensen's inequality)}\\&=1 + 2\sqrt{\Var\left[X\right]}=1+\sqrt{t}.
    \end{align*}
\end{proof}
\section{Integral upper bound for $k$-uniform hypergraphs}
\label{sec:integral_hardness}
In this section, we prove a strong upper bound against any randomized integral algorithm, showing that the greedy algorithm is almost optimal, since it achieves a competitive ratio of $1/k$.

\begin{theorem} \label{thm:integral_hardness}
For the online matching problem on $k$-uniform hypergraphs, no randomized integral algorithm can be $2/k$-competitive.
\end{theorem}
\begin{proof}
We prove that any randomized integral algorithm is at most $(2 - 2^{-k+1})/k$-competitive. To do so, we make use of Yao's principle \cite{yao1977probabilistic}: it suffices to construct a randomized instance for which any deterministic integral algorithm is at most $(2 - 2^{-k+1})/k$-competitive in expectation. Let us now describe our randomized construction $\mathcal{H} = (V,W,H)$ for any $k \in \mathbb{N}$.
\begin{itemize}
\item The offline nodes are partitioned into $k-1$ blocks: $V = C_1 \cup \dots \cup C_{k-1}$, where $|C_i| = 2(k-i)$ for each $i \in \{1, \dots, k-1\}$, meaning that $|V| = k \: (k-1)$.
\item The instance first consists of $k-1$ phases with online nodes $w_1, \dots, w_{k-1}$ arriving, all of which are incident to $2$ hyperedges. For every $i \in \{1, \dots, k-1\}$, both hyperedges incident to $w_i$ are disjoint on the offline nodes $V$ and they will both contain $(k-i)$ nodes from $C_i$, as well as $1$ node from each $C_j$ for $j \in \{1, \dots, i-1\}$.
\item We now construct a random subset of hyperedges $H_1 \subseteq H$ in the following way. At the arrival of $w_i$ for every $i \leq k-1$, pick one of the two incident hyperedges uniformly at random and put it in $H_1$. We denote by $V(H_1)$ the offline nodes spanned by $H_1$.
\item Our construction will now satisfy the following property. For every $i \in \{1, \dots, k-1\}$, after the arrival of $w_i$ and the random choice described above, we have that 
\begin{align}
\label{eq_invariant}
|C_j \setminus V(H_1)| = k - i \qquad \forall j \in \{1, \dots, i\}.
\end{align}
After a certain phase $i-1 \leq k-2$, the two hyperedges incident to $w_{i}$ in the next iteration are then both constructed as follows: take $k - i$ nodes from $C_i$ and complete it by arbitrarily picking one node from $C_j \setminus V(H_1)$ for every $j \leq i-1$.
\item After phase $k-1$, we have that $|C_j \setminus V(H_1)| = 1$ for every $j \leq k-1$, by invariant \eqref{eq_invariant}. The instance now makes one more online node $w_k$ incident to one hyperedge arrive, whose offline nodes are $C_j \setminus V(H_1)$ for every $j \in \{1, \dots, k-1\}$. Let us also add this hyperedge to $H_1$.
\end{itemize}

Let us first show that \eqref{eq_invariant} holds by induction. In the first phase, both hyperedges partition $C_1$ on the offline nodes, since $|C_1| = 2(k-1)$. One of them is chosen to enter $H_1$, meaning that $C_1\setminus V(H_1) = k-1$ after phase $1$. Let us now fix a phase $i \leq k-1$ and suppose that \eqref{eq_invariant} holds for all previous phases. By construction, $C_i$ is completely covered by the two hyperedges arriving at phase $i$, since $|C_i| = 2(k-i)$ and both of these hyperedges contain $k-i$ nodes from $C_i$. One of these hyperedges enters $H_1$ at the end of phase $i$, meaning that $|C_i \setminus V(H_1)| = k-i$ indeed holds. For any other $C_j$ with $j < i$, note that $|C_j \setminus V(H_1)| = k-i+1$ at the beginning of phase $i$, by induction hypothesis. Both hyperedges coming at phase $i$ intersect $C_j$ at two different nodes, one of which enters $V(H_1)$ by the random choice, meaning that $|C_j \setminus V(H_1)|$ drops by $1$ and equals $k-i$, thus showing \eqref{eq_invariant}.

Observe that, by construction, the hyperedges in $H_1$ are all disjoint from each other. Since we add one hyperedge to $H_1$ for every online node, we get that $\opt(\mathcal{H}) = k$. 

Let us now upper bound the value that any deterministic algorithm can get on this randomized instance. The key observation is that, if the algorithm picks a hyperedge $h \in \delta(w_i)$ which is not placed in $H_1$ for some phase $i \in \{1, \dots, k-1\}$, then it cannot pick any hyperedge arriving in later iterations. This holds, since in that case, $C_i \setminus V(H_1) \subseteq h$, and any hyperedges arriving in later iterations necessarily intersect $C_i \setminus V(H_1)$ by construction. 

Let us denote by $\val_i$ the maximum expected value that a deterministic algorithm can get if we were to start the instance from phase $i$ and go up to phase $k$. Clearly, $\val_k = 1$. For a phase $i \in \{1, \dots, k-1\}$, the algorithm can either choose not to select anything, or it picks a hyperedge and cannot pick anything in later iterations with probability $1/2$. We thus get the following recurrence relation:
\[\val_i = \max \Big\{\val_{i+1}, \frac{1}{2} + \frac{1}{2} (1 + \val_{i+1})\Big\} = \max \Big\{\val_{i+1}, 1 + \frac{1}{2}\val_{i+1}\Big\}.\]
It is easily checked that the solution to this recurrence is a geometric series $\val_{k-i} = \sum_{j = 0}^{i} 2^{-j}$ and thus $\val_1 = \sum_{j = 0}^{k-1}2^{-j} = 2 - 2^{-k+1}$. We have therefore just shown that any algorithm is at most $(2 - 2^{-k+1})/k$ competitive.
\end{proof}

\section{Rounding algorithm for online hypergraph $b$-matching}
\label{app:rounding_algo}
In this section, we consider the online $b$-matching problem on $k$-uniform hypergraphs, in which every (offline and online) node $v$ can be matched to at most $b$ hyperedges. We show that, for $b = \Omega(\log k)$, any fractional algorithm can be converted to a randomized integral algorithm while incurring a small loss in the competitive ratio.

Let $\mathcal{A}$ be a fractional algorithm that is $\rho$-competitive and let $\mathcal{H} = (V,W,H)$ be an online $k$-uniform hypergraph instance. We denote by $x \in [0,1]^H$ the fractional solution constructed by $\mathcal{A}$ on the instance $\mathcal{H}$. The rounding algorithm is now quite simple and is similar to the methods used in \cite{hutchison_b-matching_2012, raghavan_randomized_1987-2, srivastav_weighted_1995}. 

Fix some small $0 < \epsilon < 1$ and initialize two empty sets of hyperedges $S, \Match \gets \emptyset$. Upon the arrival of an online vertex $w\in W$ with $\delta(w) \subseteq H$ and $x_h \in [0,1]$ for every $h \in \delta(w)$, the rounding algorithm is as follows:
\begin{itemize}
    \item For all $h \in \delta(w)$, independently add $h$ to $S$ with probability $x'_h := (1 - \epsilon)x_h$.
    \item If $h$ was added to $S$, add it to $\Match$ as long as it does not violate the degree constraints.
\end{itemize}
The solution outputted is $\Match \subseteq H$. Let us denote this rounding algorithm by $R(\mathcal{A}, \epsilon)$. 
\begin{lemma}
Let $\mathcal{A}$ be a $\rho$-competitive fractional algorithm. The randomized integral algorithm $R(\mathcal{A}, \epsilon)$ achieves a competitive ratio of at least $(1-\epsilon)(1-k \exp(-\epsilon^2 b/3))\cdot \rho$.
\end{lemma}
\begin{proof}
    Consider an arbitrary node $v \in V \cup W$. To bound the probability that $v$ is matched to more than $b$ hyperedges in $S$, we use a Chernoff bound \cite[Theorem 1.10.1]{doerr_probabilistic_2020-1}. Fix a node $v$ and a hyperedge $h$, and let $X_{v, h}=\sum_{h'\in \delta(v)\setminus \{h\}}\mathbf{1}_{\{h'\in S\}}$. Note that $\mu:=\Exp[X_{v,h}]\leq (1-\epsilon)b$. If $\mu=0$, it is clear that $x_{v,h}\leq b$, so assume $\mu>0$. We now have:
    \begin{align*}
        \Pr\left[X_{v,h} \geq b \right]&\leq \exp\left(-\min\left(\left( \frac{b-\mu}{\mu} \right)^2,\frac{b-\mu}{\mu}\right) \mu /3\right)
        \\&\leq \exp\left(-\min\left(\frac{(b-\mu)^2}{\mu},b-\mu\right) /3\right)\\
        &\leq \exp\left(-\min\left( \epsilon^2b,\epsilon b \right) /3\right)
        \leq  \exp(-\epsilon^2 b/3),
    \end{align*}
    \todo[inline]{CK: Which version of Chernoff bound are you using (the book sits behind a paywall)? I don't think it holds when $b\gg \mu$.}
    where the second inequality follows from $b-\mu \geq \epsilon b$ and the last inequality from $b/\mu \geq 1$.
    We now upper bound the probability that a hyperedge $h$ cannot be included in $\Match$ because of the degree constraints:
    \begin{align*}
        \Pr[h\in S\setminus \Match \mid h\in S] &\leq \sum_{v\in h}\Pr\left[X_{v,h} \geq b \right]\leq k \exp(-\epsilon^2 b /3).
    \end{align*}
    Hence, we have:
    \begin{align*}
        \Exp[|\Match|] = \sum_{h\in H}\Pr[h\in \Match] &= \sum_{h\in H}\Pr[h\in S] \left(1-\Pr[h\in S\setminus \Match \mid h\in S]\right) \\ 
        &\geq \sum_{h\in H}x'_h\left(1 - k \exp(-\epsilon^2 b /3)\right)\\
        &= \left(1 - k \exp(-\epsilon^2 b /3)\right)(1-\epsilon)\sum_{h\in H}x_h \\ 
        &\geq \left(1 - k \exp(-\epsilon^2 b /3)\right)(1-\epsilon)\rho \: \optLP.
    \end{align*}\end{proof}

If $b=C\cdot \log(k)$ for some $C> 6$, then by choosing $\epsilon = \sqrt{6/C}$ we get that the competitive ratio is at least $(1-\sqrt{6/C})(1-\frac{1}{k})\rho$. \todo[inline]{CK: Did you use $\log (Ck)\leq \log C\log k$ here? If so, the lower bound of $C$ should be larger. Note also that you assumed $\epsilon<1/2$. Is this upper bound used anywhere?} By using the $\Omega(1/\log k)$-competitive fractional algorithm from \cite{buchbinder2009online}, this 
gives an $\Omega(1/\log k)$-competitive integral algorithm for this setting.

\end{document}